\documentclass{fundam}
\usepackage{url}
\usepackage{amssymb}
\usepackage{amsmath}
\usepackage[normalem]{ulem}
\usepackage{adjustbox}
\usepackage{float}
\usepackage{multirow}
\makeatletter
\newcommand{\xRightarrow}[2][]{\ext@arrow 0359\Rightarrowfill@{#1}{#2}}
\makeatother
\usepackage{tikz}
\usepackage{bussproofs}

\begin{document}


\newcommand{\overto}[1]{\stackrel{#1}{%
\overrightarrow{\smash{\,{\phantom{#1}}\,}}}}

\newcommand{\noverto}[1]{\mathrel{\hspace{.5em}\not\hspace{-.5em}\overto{#1\;}}}
\newcommand{\Rule}[2]{                                  
\frac{\raisebox{.7ex}{\normalsize{$#1$}}}
{\raisebox{-1.0ex}{\normalsize{$#2$}}}}
\newcommand{\lm}{\,\setlength{\unitlength}{1ex}\begin{picture}(1,1.75)\put(0,0){\line(1,0){1}}\put(0,0){\line(0,1){1.75}}\put(0.45,0){\line(0,1){1.75}}\end{picture}\,}
\newcommand{\llm}{\mathrel{\hspace{1mm}\rule[-0.7mm]{2mm}{.1mm}\hspace{-2mm}\parallel}}
\newcommand{\lc}{\mathrel{\hspace{1mm}\rule[-0.9mm]{2mm}{.1mm}\hspace{-2mm}|}}
\newcommand{\llc}{\mathrel{\hspace{1mm}\rule[-0.7mm]{1mm}{.1mm}\hspace{-1mm}|\hspace{0.5mm}}}

\newcommand{\conn}{\!\rightsquigarrow\!}
\newcommand{\nconn}{\!\mathrel{\hspace{.1em}\not\hspace{-.1em}\rightsquigarrow}\!}
\newcommand{\pconn}{\dashrightarrow}
\newcommand{\npconn}{\mathrel{\hspace{.1em}\not\hspace{-.1em}\dashrightarrow}}
\newcommand{\oc}{[\![}
\newcommand{\cc}{]\!]}
\newcommand{\deff}{\stackrel{\mathrm{\it def}}{=}}
\newcommand{\less}{\preccurlyeq}

\def\powerset{{I\!\!P}}
\def\nat{{I\!\!N}}
\def\dat{{I\!\!D}}
\def\snd{{\it snd}}
\def\rcv{{\it rcv}}
\def\nsnd{{\it nsnd}}
\def\nrcv{{\it nrcv}}
\def\req{{\it req}}
\def\rep{{\it rep}}
\def\sense{{\it sense}}
\def\init{{\it init}}
\def\submit{{\it submit}}
\def\deliver{{\it deliver}}
\def\data{{\it data}}
\def\rec{{\it rec}}
\def\Msg{{\it Msg}}
\title{Reliable Restricted Process Theory}


\author{Fatemeh Ghassemi\\
University of Tehran,\\ Tehran, Iran,\\
fghassemi@ut.ac.ir
\and Wan Fokkink\\
Vrije Universiteit Amsterdam,\\ Amsterdam, The Netherlands\\w.j.fokkink@vu.nl
}

\maketitle

\runninghead{F. Ghassemi, W. Fokkink}{Reliable Restricted Process Theory}

\begin{abstract}
Malfunctions of a mobile ad hoc network (MANET) protocol caused by a conceptual mistake in the protocol design, rather than unreliable communication, can often be detected only by considering communication among the nodes in the network to be reliable. In Restricted Broadcast Process Theory, which was developed for the specification and verification of MANET protocols, the communication operator is lossy. Replacing unreliable with reliable communication invalidates existing results for this process theory. We examine the effects of this adaptation on the semantics of the framework with regard to the non-blocking property of communication in MANETs, the notion of behavioral equivalence relation and its axiomatization. We illustrate the applicability of our framework through a simple routing protocol. To prove its correctness, we introduce a novel proof process, based on a precongruence relation.
\end{abstract}


\section{Introduction}
The applicability of wireless communication is growing rapidly in areas like home networks and satellite transmissions,
due to their broadcasting nature. Mobile ad hoc networks (MANETs) consist of several portable hosts
with no pre-existing infrastructure, such as routers in wired networks or access points in managed (infrastructure) wireless
networks. The design of MANET protocols is complicated, because due to mobility of nodes the topology of communication links is dynamic.
Important MANET protocols such as the Ad hoc On Demand Distance Vector (AODV) routing protocol \cite{AODV} contained flaws in their original design and have been revised accordingly.
Formal methods can be applied in the early phases of the protocol development to analyze and capture conceptual errors before their implementation. 
For instance, some errors in the design of AODV were found in \cite{spin3,LoopNamj,fehnker2013process,wrebeca} using formal techniques.

There are numerous applications of existing formal frameworks such as SPIN \cite{spin1,spin2,spin3} and UPPAAL~\cite{spin2,uppaal1,uppaal2,uppaal3,Uppaal4,GlabbeekTool} for the analysis of MANET protocols. Lack of support for compositional modeling and arbitrary topology changes motivates developing a new approach, tailored to the domain of MANETs, with a primitive for local broadcast and supporting the verification of MANET protocols against changes of the underlying topology. The tailored formal modeling framework should provide some form of wireless communication which varies at the different layers of the Open Systems Interconnection (OSI) model: physical, data link, network, transport, session, presentation, and application. 
For instance, the data link layer is responsible for transferring data across the physical link and handling conflicts due to simultaneous accesses to the shared media. In contrast, communication at the network layer provides point-to-point communication between two nodes that are not directly connected through appropriate routing of messages by using the communication service of the data link layer.
Most frameworks for the formal analysis of MANET protocols, such as ~\cite{CBSNanz,CMAN,Godskesen08,CMN,bKlaim,w-cal,GodskesenN09,CDST,GlabbeekAWN,wrebeca}, focus on protocols above the data link layer; hence they support the core services of this layer, which means that local broadcast is the primitive means of communication. Wireless communication at this layer is \emph{non-blocking}, i.e., the sender broadcasts irrespective of the readiness of its receivers, and is \emph{asynchronous}, i.e., received packets are buffered at the receiver. The data link layer of a node processes the packet if it is an intended destination.  While a node is busy processing a message, it can still receive messages, buffer them and process them later. However, if two different nodes broadcast simultaneously with a common node in their range, the latter node cannot receive both messages and drops one of them, which is called the hidden node problem. We say that wireless communication is \emph{reliable} if the intended receivers successfully receive the packet. In other words, message delivery is guaranteed to all connected neighbors.

Although lossy communication is an integral part of MANETs, mimicking it faithfully in a formal framework can hamper the formal analysis of MANET protocols. To obtain a deeper understanding of a malfunctioning of such a protocol due to a conceptual mistakes in its design rather than unreliable communication, it may be helpful to consider communication reliable, meaning that the possibility of the hidden node problem is omitted from the framework. Therefore we introduced the process algebra {\it Reliable Restricted Broadcast Process Theory (RRBPT)} in \cite{FMSD16}, to perform model checking of MANET protocols in a setting where communication is reliable. It is a variant of Restricted Broadcast Process Theory {\it (RBPT)} that we introduced previously in \cite{Fatemeh08} for the modeling and analysis of protocols above the data link layer. The underlying semantic model of {\it RBPT}, a so-called constrained labeled transition system (CLTS), implicitly considers mobility of nodes with the novel notion of a network constraint, which abstractly defines a set of topologies: those satisfying the given connectivity constraints. The transitions of a CLTS are annotated with appropriate network constraints to restrict the behavior to MANETs with a topology of the specified ones. {\it RBPT} was extended with a set of auxiliary operators to reason about MANETs by equational reasoning, so-called Computed Network Process Theory ({\it CNT}) \cite{FatemehFI10}. We provided a sound and complete axiomatization for {\it CNT} terms with finite-state behaviors, modulo so-called rooted branching computed network bisimilarity. This axiomatization enables linearization of processes at the syntactic level to take advantage of symbolic verification~\cite{grootefoci,WanPangPolFoci}, especially when the network is composed of similar nodes \cite{GWParallel,FatemehTCS}.

Somewhat surprisingly, all these results do not carry over in a
straightforward fashion from {\it RBPT} to {\it RRBPT}. To put the
model checking approach presented in \cite{FMSD16} on a firm basis,
the current paper develops the formal foundations for {\it RRBPT}
and modifies the core of {\it CNT}. In a lossy setting, the
non-blocking property of local broadcast communication is an
immediate consequence of the rule ${\it Par}$ and its counterpart
for the parallel composition: $\Rule{t_1\overto{a}t_1'}{t_1\parallel
t_2\overto{a}t_1'\parallel t_2}$, which expresses that if a node is
not ready to participate in a communication, then we can assume that
either it was disconnected from the sender or it was connected but
has lost the message. However, in the reliable setting, to guarantee
the  non-blocking property, nodes should always be input-enabled.
{\it RRBPT} provides a sensing operator which allows to change the
control flow of a process depending on the status of node
connectivity with other nodes. The input-enabledness feature is
ensured through the {\it RRBPT} operational rules, where the main
difference between {\it RRBPT} and {\it RBPT} is: in {\it RRBPT},
nodes lose a communication only when they are disconnected and are
always input-enabled. We recap challenges of bringing
input-enabledness feature in the semantics of {\it RRBPT} in the
presence of the sensing operator. Furthermore, the behavioral
equivalence relation of {\it CNT} setting is not a congruence with
respect to parallel composition anymore. To support the desired
distinguishing power, we provide a new bisimulation relation which
guarantees the congruence property for MANETs. {\it RRBPT} can be
extended in the same way as {\it RBPT} with computed network terms
and the auxiliary operators \emph{left merge} ($\lm$) and
\emph{communication merge} ($\mid$) to provide a sound and complete
axiomatization for the parallel composition. However, the
input-enabledness feature and the new sensing operator require new
auxiliary operators to assist their axiomatization. To this aim, we
discuss the appropriate axioms of {\it RRBPT}. We utilize our axioms
to analyze the correctness of protocols at the syntactic level. To
this aim, we facilitate the specification of the protocol behaviors
preconditioned to multihop constraints and then introduce a new
notion of refinement among protocol implementations and their
specifications. We demonstrate the applicability of our framework by
analyzing and proving the correctness of a simple routing protocol
inspired by the AODV protocol.



This paper is organized as follows. Sections \ref{sec::CLTS} and \ref{sec::RRBPTsemantics} introduce our semantic model and explain how it is helpful in giving semantics to reliable communication. Section \ref{sec::syntaxsemantics} introduces the syntax of {\it RRBPT}. Sections \ref{sec::reliablesim} and \ref{sec::raxiom} provide the appropriate notion of behavioral equivalence and axioms in the reliable setting, respectively. We demonstrate the applicability of our new framework by analyzing a simple routing protocol in Section \ref{sec::casestudy}.
We review and compare the related process algebraic frameworks in depth in Section \ref{sec::concreteRelatedWork} before concluding the paper.

\section{\label{sec::CLTS}Constrained Labeled Transition Systems}
Let ${\it Loc}$ denote a set of network addresses, ranged over by $\ell$. Viewing a network topology as a directed graph, it can be defined as $\gamma:{\it Loc}\rightarrow \powerset({\it Loc})$, where $\gamma(A)$ expresses the set of nodes that are directly connected to $A$, and hence, can receive message from $A$.  A network constraint $\mathcal{C}$ is a set
of connectivity pairs $\conn\,: {\it Loc}\times {\it Loc}$ and
disconnectivity pairs $\nconn\,: {\it Loc}\times {\it Loc}$. In this
setting, non-existence of (dis)connectivity information between two
addresses implies lack of information about this link (which can
e.g.\ be helpful when the link has no effect on the evolution of the
network). For instance, $B\conn A$ denotes that $A$ is connected to
$B$ directly and consequently $A$ can receive data sent by $B$ as before,
while $B\nconn A$ denotes that $A$ is not connected to $B$ directly
and consequently cannot receive any message from $B$. The direction of an arrow shows the direction of information flow. We write
$\{B\conn A,C,~B\nconn D,E\}$ instead of $\{B\conn A,~B\conn
C,~B\nconn D,~B\nconn E\}$. The set ${\it Loc}$ is
extended with the unknown address $?$ to represent the
address of a node which is still not known or concealed from an
external observer. For instance, the leader
address of a node can be initialized to this value.
Furthermore, to define the semantics of communicating nodes in terms of
restrictions over the topology in a compositional way, the semantics of
receive actions can be defined through an unknown sender, which will be
replaced by a known address when the receive actions are composed with the corresponding send
action at a specific node (see Section \ref{sec::RRBPTsemantics}).

A network
constraint $\mathcal{C}$ is said to be \emph{well-formed} if $\forall \ell,\ell'\in{\it
Loc}~(\ell\conn\ell'\not\in\mathcal{C}\vee\ell\nconn\ell'\not\in\mathcal{C})
$. 
Let $\mathbb{C}^v({\it Loc})$ denote the
set of well-formed network constraints that can be defined over
the network addresses in ${\it Loc}$. We define an ordering on network constraints. We say that $\mathcal{C}_1 \preccurlyeq \mathcal{C}_2$ iff $\mathcal{C}_2\subseteq \mathcal{C}_1$ or $\exists\, \ell\in {\it Loc}\,(\mathcal{C}_2[\ell/?]\subseteq \mathcal{C}_1)$, where $d[d_1/d_2]$
denotes the substitution of $d_1$ for $d_2$ in $d$; this can
be extended to process terms. For instance, $\{B\conn A\}\preccurlyeq \{?\conn A\}$ and $\{B\conn A,\,B\conn C\}\preccurlyeq \{B\conn A\}$. Each well-formed network
constraint $\mathcal{C}$ represents the set of network topologies
that satisfy the (dis)connectivity pairs in $\mathcal{C}$, i.e., $\Gamma(\mathcal{C})=\{
\gamma \mid \mathcal{C}_{\Gamma}(\gamma) \preccurlyeq \mathcal{C} \}$
where $\mathcal{C}_{\Gamma}(\gamma)=\{\ell\conn\ell'\mid
\ell'\in\gamma(\ell)\}\cup\{\ell\nconn\ell'\mid
\ell'\not\in\gamma(\ell)\}$ extracts all one-hop (dis)connectivity
information from $\gamma$. So the empty network constraint $\{\}$ still
denotes all possible topologies over ${\it Loc}$. The negation
$\neg\,\mathcal{C}$ of network constraint $\mathcal{C}$ is obtained
by negating all its (dis)connectivity pairs. Clearly, if
$\mathcal{C}$ is well-formed then so is $\neg\,\mathcal{C}$.

\textit{Constrained labeled transition system}s (CLTSs) provide a semantic model for the operational
behavior of MANETs. Let ${\it Msg}$ denote a
set of messages communicated over a network and ranged over by
$\mathfrak{m}$. Let ${\it Act}$ be the network send and receive
actions with signatures $\nsnd:{\it Msg}\times {\it Loc}$
and $\nrcv:{\it Msg}$, respectively. The send action
$\nsnd(\mathfrak{m},\ell)$ denotes that the message $\mathfrak{m}$
is transmitted from a node with the address $\ell$, while the
receive action $\nrcv(\mathfrak{m})$ denotes that the message
$\mathfrak{m}$ is ready to be received. Let ${\it Act}_\tau={\it
Act}\cup\{\tau\}$, ranged over by $\eta$.

\begin{definition}
A \emph{CLTS} is a tuple $\langle S,\Lambda,\rightarrow,s_0\rangle$,
with $S$ a set of states,  $\Lambda\subseteq \mathbb{C}^v({\it Loc})\times {\it
Act}_\tau$, $\rightarrow\,\subseteq S\times \Lambda\times S$ a
transition relation, and $s_0\in S$ the initial state. A transition
$(s,(\mathcal{C},\eta),s')\in\,\rightarrow$ is denoted by
$s\overto{(\mathcal{C},\eta)}s'$.
\end{definition}

Generally speaking, the transition $s\overto{(\mathcal{C},\eta)}s'$ expresses that a MANET protocol in
state $s$ with an underlying topology $\gamma\in\Gamma(\mathcal{C})$ can
perform action $\eta$ to evolve to state $s'$.

The semantics of broadcast communication is defined to be reliable if and if only the nodes that are connected to the sender, as defined by its corresponding network constraint, receive the message. 
We remark that the status of the links from the receivers to the sender or between two arbitrary receivers are not of importance and hence, they are abstracted away. Therefore, by constructing such network constraints through the semantic rules, reliable communication is brought into our framework.

\section{Syntax of {\it RRBPT}}\label{sec::syntaxsemantics}
Let $\mathcal{A}$ denotes a countably infinite set of
process names which are used as recursion variables in recursive
specifications. Besides \textit{network send} and \textit{receive} actions, i.e., $\nsnd(\mathfrak{m},\ell)$ and $\nrcv(\mathfrak{m})$, we assume \textit{protocol send} and \textit{receive} actions, denoted by
$\snd,\rcv:{\it Msg}$, i.e., parametrized by messages.  Furthermore, let
${\it IAct}$ be a set of internal actions. The syntax of
{\it RRBPT} is given by the following grammar:
\[ \begin{array}{l}  t:: = ~  0 ~~ |~~ \alpha.t  ~~ |~~ t+t~~|~~ \oc  t\cc_\ell~~ |~~ t\parallel
t~~|~~\mathfrak{A},\,\mathfrak{A}\deff t ~~ | ~~ {\it sense}(\ell,t,t) ~~|~~(\nu\ell) t  ~~|~~\tau_m(t)~~|~~\partial_m(t)
\end{array} \]
The deadlock process is modeled by $0$. The process $\alpha.t$
performs action $\alpha$ and then behaves as process $t$, where
$\alpha$ is either an internal action or a protocol send/receive
action $\snd(\mathfrak{m})/\rcv(\mathfrak{m})$. 
Internal actions are
useful in modeling the interactions of a process with other
applications running on the same node. Protocol send/receive actions
specify the interaction of a process with its data-link layer
protocols: these protocols are responsible for transferring messages
reliably throughout the network. These actions are turned into their
corresponding network ones via the semantics (see Section
\ref{sec::RRBPTsemantics}). 
The process $t_1+t_1$ behaves non-deterministically as $t_1$ or
$t_2$. The simplest form of a MANET is a node, represented by the network
deployment operator $\oc t\cc_\ell$, denoting process $t$ deployed
on a node with the known network address $\ell\neq\,?$ (where $?$
denotes the unknown address). A MANET can be
composed by putting MANETs in parallel using $\parallel$; the nodes communicate with
each other by reliable restricted broadcast. A process name is
specified by a recursive equation $\mathfrak{A}\deff t$ where $\mathfrak{A}\in\mathcal{A}$ is a name.

MANET protocols may behave based on the (non-)existence of a link.
A neighbor discovery service can be implemented at the data link layer,
by periodically sending \textit{hello} messages and acknowledging
such messages received from a neighbor. The sensing operator ${\it sense}(\ell',t_1,t_2)$
examines the status of the link from the node, say with address $\ell$, that the sensing is executed
on to the node with the address $\ell'$;  in case of its existence it behaves as $t_1$, and otherwise as
$t_2$. For instance, the term $\oc {\it
sense}(\ell',t_1,t_2)\cc_\ell$ examines the existence of the link
$\ell\conn \ell'$, and then behaves accordingly. As a running example, $P \deff \sense(B,\snd(\data_B).P,0)$ denotes a process that recursively broadcasts a data
message $\data_B$ as long as it is connected to $B$; 
and $Q\deff 
\rcv(\data_B).\deliver.Q$
a process that recursively receives a data message $\data$ and then
the internal action $\deliver$ upon successful receipt of data. 
The network process $\oc P\cc_A\parallel
\oc Q\cc_B$ specifies an ad hoc network composed of two nodes with the
network addresses $A$ and $B$ deploying processes $P$ and $Q$,
respectively.

The hide operator
$(\nu\ell)t$ conceals the address $\ell$ in the process $t$, by
renaming this address to $?$ in network send/receive actions. For
each message $m\in {\it Msg}$, the abstraction operator $\tau_m(t)$
renames network send/receive actions over messages of type $m$ to
$\tau$, and the encapsulation operator $\partial_m(t)$ forbids
receiving messages of type $m$. Let $\tau_{\{m_1,\ldots,m_n\}}(t)$
and $\partial_{\{m_1,\ldots,m_n\}}(t)$ denote
$\tau_{m_1}(\ldots(\tau_{m_n}(t))\ldots)$ and
$\partial_{m_1}(\ldots(\partial_{m_n}(t))\ldots)$.

For example, {\small $\tau_{\it Msg}(\partial_{\it Msg}(\oc
P\cc_A\parallel  \oc Q\cc_B))$} specifies an isolated MANET that
cannot receive any message from the environment, while its
communications (i.e., send actions) are abstracted away.

Terms should be grammatically well-defined, meaning that processes
deployed at a network address are only defined by action prefix,
choice, sense and process names. Furthermore, the application of action prefix,
choice, sense and process names is restricted to the deployment operator.

\section{\label{sec::RRBPTsemantics}Semantics of {\it RRBPT}}
The operational rules in
Table 
\ref{Tab::RRBPTsemanticsII}
induce a CLTS with transitions of the form $t\overto{\beta} t'$,
where $\beta\in\mathbb{C}^v({\it Loc})\times{\it Act}_\tau$ where ${\it Act}=\{{\it NAct}\cup{\it
IAct}\}$, ${\it NAct}$ denotes the set of network
send and receive actions, and ${\it IAct}$ the set of internal
actions ranged over by $i$. Assume that $\alpha$ denotes actions of the form $\{\rcv(\mathfrak{m}),\snd(\mathfrak{m})\,\mid\, \mathfrak{m}\in{\it Msg}\}$. In these rules,
$t\noverto{(\mathcal{C},~\nrcv(\mathfrak{m}))}$ denotes that there
exists no $t'$ such that
$t\overto{(\mathcal{C}',~\nrcv(\mathfrak{m}))}t'$ and
$\mathcal{C}'\preccurlyeq \mathcal{C}$. The symmetric counterparts of the rules $\it Choice$, $\it Bro$, and $\it Par$ hold, but have been omitted for the brevity. 

\begin{table*}[htbp]\caption{Semantics of {\it RRBPT} operators.}\label{Tab::RRBPTsemanticsII}
\centering\begin{tabular}{c}
$\Rule{ t_1\overto{(\mathcal{C},\alpha)}
t_1'}{ {\it sense}(\ell,t_1,t_2)\overto{(\{\ell\,\conn\,
~?\}\cup\mathcal{C},\alpha)}
t_1'}:{\it Sen}_1$~~~~$\Rule{~}{\alpha.t\overto{(\{\},\alpha)}t}:{\it
Prefix}$\vspace{4mm}\\
$\Rule{
t_2\overto{(\mathcal{C},\alpha)} t_2' }{  {\it
sense}(\ell,t_1,t_2)\overto{(\{\ell\,\nconn\,
~?\}\cup\mathcal{C},\alpha)} t_2'}:{\it Sen}_2$~~~~$\Rule{t\overto{\beta}
t'}{(\nu\ell)t\overto{\beta[?/\ell]} (\nu\ell)t'}:{\it Hid}$\vspace{4mm}\\
$\Rule{t\overto{(\mathcal{C},~\snd(\mathfrak{m}))}t'}{\oc
t\cc_\ell\overto{(\mathcal{C}[\ell/?],~\nsnd(\mathfrak{m},\ell))} \oc
t'\cc_\ell}:{\it
Snd}$~~~
$\Rule{t\overto{(\mathcal{C},~\rcv(\mathfrak{m}))}t'}{\oc t\cc_\ell
\overto{(\mathcal{C}[\ell/?]\cup\{?\,\conn\,
\ell\},~\nrcv(\mathfrak{m}))} \oc t'\cc_\ell}:{\it
Rcv}_1$\vspace{4mm}\cr
$\Rule{
\oc t\cc_\ell \noverto{(\mathcal{C},~\nrcv(\mathfrak{m}))}
~~~\nexists \mathcal{C}'(\oc t\cc_\ell \noverto{(\mathcal{C}',~\nrcv(\mathfrak{m})}\,\wedge \, \mathcal{C}\preccurlyeq \mathcal{C}')}{\oc
t\cc_\ell \overto{(\mathcal{C},~\nrcv(\mathfrak{m}))} \oc t\cc_\ell}:{\it
Rcv}_2$\vspace{4mm}\cr
$\Rule{t_1\overto{\beta}t_1'}{t_1+t_2\overto{\beta}t_1'}:{\it Choice}$~~~            $\Rule{t_1\overto{(\mathcal{C}_1,~\nsnd(\mathfrak{m},\ell))}t_1' ~~
t_{2i}
\overto{(\mathcal{C}_2,~\nrcv(\mathfrak{m}))}t_{2}'}{t_1
\parallel t_2\overto{(\mathcal{C}_1 \cup\,\mathcal{C}_2[\ell/?],~\nsnd(\mathfrak{m},\ell))}
t_1'\parallel t_2'}:{\it Bro}$
\vspace{4mm}\cr
$\Rule{t\overto{\beta} t'}{\mathfrak{A}\overto{\beta} t'}:~{\it
Inv},~\mathfrak{A}\deff t$~~~$\Rule{t_1\overto{(\mathcal{C}_1,\nrcv(\mathfrak{m}))}t_1' ~~ t_2 \overto{(\mathcal{C}_2,\nrcv(\mathfrak{m}))}t_2'}{t_1 \parallel
t_2\overto{(\mathcal{C}_1\cup \mathcal{C}_2,\nrcv(\mathfrak{m}))}t_1'\parallel t_2'}:{\it Recv}$\vspace{4mm}\cr
$\Rule{t\overto{(\mathcal{C},\eta)}t'}{t
\overto{(\mathcal{C}',\eta)} t'}:~{\it Exe},~ \mathcal{C}'  \preccurlyeq \mathcal{C}$~~~$\Rule{t\overto{(\mathcal{C},~i)}t'}{\oc
t\cc_\ell\overto{(\mathcal{C},~i)} \oc
t'\cc_\ell}:{\it
Int}$~~~~~~$\Rule{t_1\overto{(\mathcal{C},\eta)}
t_1'~~~\eta\in{\it IAct}\cup\{\tau\}}{t_1
\parallel t_2\overto{(\mathcal{C},\eta)} t_1'\parallel
t_{2}}:{\it
Par}$\vspace{4mm}\cr
$\Rule{t\overto{(\mathcal{C},\eta)}t'~~~\eta\neq\nrcv(m) }{\partial_{m}(t)\overto{(\mathcal{C},\eta)}\partial_{m}(t')}:{\it
Encap}$~~~~~
$\Rule{t\overto{(\mathcal{C},\eta)} t'}{\tau_{m}(t)\overto{(\mathcal{C},\tau_m(\eta))}\tau_{m}(t')}:{\it
Abs}$
\end{tabular}\end{table*}

Rule $\it Prefix$ assigns an empty network constraint to each prefixed action, which may be accumulated by further constraints through application of rules ${\it Rcv}_1$ or ${\it Sen}_{1,2}$. 
The rule ${\it Int}$ indicates that a node progresses when the deployed process on the node performs an internal action.
Interaction between the process $t$ and its data-link layer is
specified by the rules ${\it Snd}$ and ${\it Rcv}_{1,2}$: when $t$ broadcasts a message, it
is delivered to the nodes in its transmission range, disregarding
their readiness. ${\it Rcv}_1$ specifies that a process $t$ with an
enabled receive action can perform it successfully if it has a link
to a sender (not currently known).  
If a node does not have any
enabled receive action $\nrcv(\mathfrak{m})$ for the network constraint $\mathcal{C}$, then receiving the
message has no effect on the node behavior, as explained by ${\it
Rcv}_2$. This rule also implicitly implies that an enabled receive action cannot be performed when the node is disconnected from the
sender (not currently known). Consequently, this rule makes nodes
\textit{input-enabled}, meaning that a node not ready to receive a
message will drop it. 
Rule ${\it Rcv}_2$ adds a network receive action $(\mathcal{C},\nrcv(\mathfrak{m}))$ to the behavior of a network node, specified by $\oc t\cc_\ell$,  if it has no transition $(\mathcal{C}',\nrcv(\mathfrak{m}))$ such that $\mathcal{C}'\less \mathcal{C}$. Furthermore, this rule ensures that a most general $\mathcal{C}$ is selected, and hence, the receive action $\nrcv(\mathfrak{m})$ is defined for all possible network constraints (when combined with rule $\it Exe$). Therefore, $\oc P\cc_A$ has a $(\{\},\nrcv(\data_B))$-transition by application of this rule.

Rules ${\it Sen}_{1,2}$ explain the behavior of the sense operator. In case there is
a link to the node with the address $\ell$ from the node that is running
the sense operator, and currently its address is unknown, then it behaves
like $t_1$; in case this link is not present, it behaves like $t_2$. Therefore, the link status is combined
with the network constraint $\mathcal{C}$ generated by its first or
second term argument, as given by ${\it Sen}_{1,2}$ respectively.
For instance, by $\it Prefix$ and ${\it Sen}_1$, $P$ only
generates a $(\{?\,\conn\,B\},~\snd(\data_B))$-transition.

In rules $\it Snd$ and ${\it Rcv}_{1}$, the network constraint
$\mathcal{C}$ may have the unknown address due to sensing
operators, which is replaced by the address of the deployment
operator, i.e., $\mathcal{C}[\ell/?]$.
Therefore, by applying $\it
Snd$ to the only transition of $P$,  $\oc P\cc_A$ generates a $(\{A\conn B\},\nsnd(\data_B))$-transition.

Rule ${\it Recv}$ synchronizes the receive actions of processes
$t_1$ and $t_2$ on message $\mathfrak{m}$, while combining together
their (dis)connectivity information in network constraints
$\mathcal{C}_1$ and $\mathcal{C}_2$. Rule ${\it Bro}$ specifies
how a communication occurs between a receiving and a sending
process. This rule combines the network constraints, while the
unknown location (in the network constraint of the receiving
process) is replaced by the concrete address of the sender. In
${\it Bro}$ and ${\it Recv}$ it is required that the union of network
constraints on the transition in the conclusion be well-formed.

The rule $\it Par$ prevents evolution of sub-networks
on network actions, in contrast  to lossy settings, and enforces all nodes to specify their
localities with respect to the sender before evolving the whole
network via ${\it Recv}$ or ${\it Bro}$ rules. It only allows a process to evolve by performing an internal or silent action. $\it Exe$ explains that a behavior
that is possible for a network constraint, is also possible for a
more restrictive network constraint. 

For instance, the MANET $\oc P\cc_A
\parallel \oc Q\cc_B$ can generate the $(\{B\conn A\},~\nsnd({\it
data}_B,A))$ transition induced by the deduction tree below, where
$y\equiv {\it deliver}.Q$: {\small\begin{prooftree}
\AxiomC{}\RightLabel{:$\it Prefix$}\UnaryInfC{$P\overto{(\{\},~\snd({\it data}_B)) } P$}\RightLabel{:${\it
Sen}_1$} \UnaryInfC{$P\overto{(\{?\,\conn\,B\},~\snd({\it
data}_B))}P$}\RightLabel{:${\it Snd}$}\UnaryInfC{$\oc
P\cc_A\overto{(\{A\,\conn\,B\},~\nsnd({\it data}_B,A))}\oc P\cc_A$}
\AxiomC{}\RightLabel{:$\it
Prefix$}\UnaryInfC{$Q\overto{(\{\},~\rcv({\it data}_B))}
y$}\RightLabel{:${\it Rcv}_1$} \UnaryInfC{$\oc  Q\cc_B
\overto{(\{?\,\conn\,B\},~\nrcv({\it data}_B))}\oc y\cc_B$}\RightLabel{:${\it
Bro}$} \BinaryInfC{$\oc P\cc_A\parallel \oc  Q\cc_B
\overto{(\{B\,\conn\,A\},~\nsnd({\it data}_B,A))}\oc P\cc_A\parallel \oc
y\cc_B$}\end{prooftree}}

Rule $\it Hid$ replaces every occurrence of $\ell$ in the network
constraint and action of $\beta$ by $?$, and hence hides activities
of a node with address $\ell$ from external observers. According to
$\it Abs$, the abstraction operator $\tau_m$ converts all network
send and receive actions with a message of type $m$ to $\tau$ and
leaves other actions unaffected, as defined by the function
$\tau_m(\eta)$. The encapsulation operator $\partial_m$ disallows
all network receive actions on messages of type $m$, as specified by
$\it Encap$.

The semantics of {\it RRBPT} was first introduced in \cite{FMSD16} with the aim of defining CLTSs with negative connectivity pairs to illustrate their benefit for model checking MANET protocols. In this research, we modify its semantics to properly define the behavior of MANETs in the reliable setting. To this end, two groups of rules have been modified substantially: those of receive actions and the sensing operator. More specifically, the operational semantics of receive action  in \cite{FMSD16} 
explicitly specifies the
locality of the receiver node with respect to the sender (that could be  connected,
disconnected, or unknown) through three semantic rules. Furthermore, the semantics of the sensing operator in \cite{FMSD16} makes $\oc P\cc_A$ move by $(\{B\nconn A,?\conn A\},\nrcv(\data_B))$ and $(\{B\nconn A,?\nconn A\},\nrcv(\data_B))$ to $\oc 0\cc_A$ while here it has a self-loop with the label of $(\{B\nconn A\},\nrcv(\data_B))$. In other words, the chance of sending $\data_B$ is lost after dropping a received message of $\data_B$. Such a drawback is resolved by the newly introduced rule ${\it Rcv}_2$ and removing two previous rules of the sensing operator. 

%

\section{Rooted Branching Reliable Computed Network Bisimilarity}\label{sec::reliablesim}
Terms of the lossy framework {\it RBPT} are considered modulo rooted branching computed
network bisimilarity~\cite{FatemehFI10}.
This equivalence relation is defined using the following notations:

\begin{itemize}
\item $\Rightarrow$ denotes the reflexive and transitive closure of unobservable actions: \begin{itemize}
\item $t\Rightarrow t$;
\item if $t\overto{(\mathcal{C},\tau)}t'$ for some arbitrary network constraint $\mathcal{C}$ and
$t'\Rightarrow t''$, then $t\Rightarrow t''$.
\end{itemize}
\item $t\overto{\langle (\mathcal{C},\eta)\rangle}t'$ iff $t\overto{(\mathcal{C},\eta)}t'$ or $t\overto{(\mathcal{C}[\ell/?],\eta[\ell/?])}t'$ and
$\eta$ is of the form $\nsnd(\mathfrak{m},{?})$ for some $\mathfrak{m}$.
\end{itemize}
Intuitively $t\Rightarrow t'$ expresses that after a number of communications, $t$ can behave like $t'$. Furthermore, an action like
$(\{{?}\conn B\},\nsnd({\it req}({?}),{?}))$ can be matched to an action like
$(\{A\conn B\},\nsnd({\it req}(A),A))$, which is its $\langle -
\rangle$ counterpart.

\begin{definition} {\label{Def::bbism}} A binary relation
$\mathcal{R}$ on {\it RBPT} terms is a branching computed network
simulation if $t_1\mathcal{R} t_2$ and $t_1\overto{(\mathcal{C},\eta)} t_1'$ implies that either:
\begin{itemize}
\item $\eta$ is of the form $\nrcv(\mathfrak{m})$ or $\tau$, and $t_1'\mathcal{R} t_2$; or
\item there are $t_2'$ and $t_2''$ such that $t_2\Rightarrow t_2''\overto{\langle(\mathcal{C},\eta)\rangle} t_2'$, where  $t_1\mathcal{R} t_2''$ and $t_1'\mathcal{R}
t_2'$.
\end{itemize} $\mathcal{R}$ is a branching computed network bisimulation if $\mathcal{R}$ and ${\mathcal{R}}^{-1}$
are branching computed network simulations. Two terms $t_1$ and
$t_2$ are branching computed network bisimilar, denoted by $t_1
\simeq_{b} t_2$, if $t_1\mathcal{R} t_2$ for some branching computed
network bisimulation relation $\mathcal{R}$.
\end{definition}

This definition distinguishes process terms according to their abilities to broadcast messages, and therefore, MANET protocols that can only receive are treated as deadlock as they cannot send any observable message.

\begin{definition} {\label{Def::rbbism}} Two terms
$t_1$ and $t$ are \emph{rooted branching computed network
bisimilar}, written $t_1 \simeq_{rb} t_2$, if:
\begin{itemize}
\item $t_1 \overto{(\mathcal{C},\eta)}t_1'$ implies there is a $t_2'$ such that
$t_2\overto{\langle(\mathcal{C},\eta)\rangle}t_2'$ and $t_1'\simeq_{b} t_2'$;
\item $t_2 \overto{(\mathcal{C},\eta)}t_2'$ implies there is a $t_1'$ such that
$t_1\overto{\langle(\mathcal{C},\eta)\rangle} t_1'$ and $t_1'\simeq_{b} t_2'$.
\end{itemize}
\end{definition}

Rooted branching computed network bisimilarity does not constitute a congruence with respect to the {\it RRBPT} operators. We still  want that a
receiving MANET (after its first action) be equivalent to
deadlock. In this setting, still $\oc 0\cc_A\simeq_b \oc \rcv(\mathfrak{m}).0\cc_A$, but $\oc 0\cc_A\parallel \oc \snd(\mathfrak{m}).0\cc_B \not\simeq_b \oc \rcv(\mathfrak{m}).0\cc_A\parallel \oc \snd(\mathfrak{m}).0\cc_B$, since by application of ${\it Rcv}_{1,2}$, $\it Snd$, and ${\it Bro}$:\[\begin{array}{l}
\oc \rcv(\mathfrak{m}).0\cc_A\parallel \oc \snd(\mathfrak{m}).0\cc_B\overto{(\{B\nconn A\},\nsnd(\mathfrak{m},B))}\oc \rcv(\mathfrak{m}).0\cc_A\parallel \oc 0\cc_B\\
\oc \rcv(\mathfrak{m}).0\cc_A\parallel \oc \snd(\mathfrak{m}).0\cc_B\overto{(\{B\conn A\},\nsnd(\mathfrak{m},B))}\oc 0\cc_A\parallel \oc 0\cc_B
\end{array}\]while by application of ${\it Rcv}_2$, $\it Snd$, ${\it Bro}$: \[\oc 0\cc_A\parallel \oc \snd(\mathfrak{m}).0\cc_B\overto{(\{\},\nsnd(\mathfrak{m},B))}\oc 0\cc_A\parallel \oc 0\cc_B\]which cannot be matched to any transition of $\oc \rcv(\mathfrak{m}).0\cc_A\parallel \oc \snd(\mathfrak{m}).0\cc_B$ according to the second condition of Definition \ref{Def::bbism}. However, we observe that the $(\{\},\nsnd(\mathfrak{m},B))$-transition can be matched to the transition sets of actions $(\{B\nconn A\},\nsnd(\mathfrak{m},B))$ and $(\{B\conn A\},\nsnd(\mathfrak{m},B))$, as the network constraints $\{B\nconn A\}$ and $\{B\conn A\}$ provide a partitioning of $\{\}$ while the resulting states of their corresponding transitions are equivalent. Thus, we revise our Definition \ref{Def::bbism} by generalizing its second condition. 

Intuitively, two MANETs are equivalent if they have the same observable behaviors for all possible underlying topologies. In the lossy setting, the observable behaviors exclude receive actions, as the node $\oc \rcv(a).\snd(a).0\cc_A$ can be distinguished from $\oc \rcv(a).0\cc_A$ due to its capability to send $a$ after its receipt. However, the capability of receiving messages implicitly defines a restriction on the underlying topology. For instance, the sending action $\snd(a)$ in $\oc \rcv(a).\snd(a).0\cc_A$ is only possible if the node in question was previously connected to a sender and successfully received $a$. Thus to distinguish $\oc \rcv(a).\snd(a).0\cc_A$ from $\oc\snd(a).0\cc_A$, receive actions are included in the observables in the reliable setting. Furthermore, as dropping a message may have the same effect as its processing (as explained above), a transition cannot be matched in the same way as in Definition \ref{Def::bbism} and it may be matched to multiple transitions. A partitioning of a network constraint $\mathcal{C}$ consists of network constraints $\mathcal{C}_1$, $\ldots$ , $\mathcal{C}_n$ such that 
$\forall i,j\le n\,(i\neq j\Rightarrow \Gamma(\mathcal{C}_i)\cap\Gamma(\mathcal{C}_j)=\emptyset)\,\wedge \, \bigcup_{k=1}^n \Gamma(\mathcal{C}_k) = \Gamma(\mathcal{C})$.

\begin{definition} {\label{Def::brbism}} A binary relation
$\mathcal{R}$ on {\it RRBPT} terms is a branching reliable computed network
simulation if $t_1~\mathcal{R}~ t_2$ and $t_1\overto{(\mathcal{C},\eta)} t_1'$ imply that either:
\begin{itemize}
\item $\eta$ is a $\tau$ action, and $t_1'~\mathcal{R}~ t_2$; or
\item there are $s_1'',\ldots,s_k''$ and $s_1',\ldots,s_k'$ for some $k>0$ such that $\forall i\le k(t_2\Rightarrow s_i''\overto{\langle(\mathcal{C}_i,\eta)\rangle} s_i'$, with  $t_1~\mathcal{R}~ s_i''$ and  $t_1'~\mathcal{R}~
s_i')$, and $\langle\mathcal{C}_1\rangle,\ldots,\langle\mathcal{C}_k\rangle$ constitute a partitioning of $\langle\mathcal{C}\rangle$.
\end{itemize} $\mathcal{R}$ is a branching reliable computed network bisimulation if $\mathcal{R}$ and ${\mathcal{R}}^{-1}$
are branching reliable computed network simulations. Two terms $t_1$ and
$t_2$ are branching reliable computed network bisimilar, denoted by $t_1
\simeq_{br} t_2$, if $t_1~\mathcal{R~} t_2$ for some branching reliable computed    network bisimulation relation $\mathcal{R}$.
\end{definition}

Trivially $(t_1\simeq_{b} t_2)\,\Rightarrow\, (t_1\simeq_{br}t_2)$.

\begin{theorem}\label{The::eqiv}
Branching reliable computed network bisimilarity is an equivalence.
\end{theorem}

See Section \ref{sec::app} for the proof of this theorem.

\begin{definition} {\label{Def::rbrbism}} Two terms
$t_1$ and $t$ are \emph{rooted branching reliable computed network
bisimilar}, written $t_1 \simeq_{rbr} t_2$, if:
\begin{itemize}
\item $t_1 \overto{(\mathcal{C},\eta)}t_1'$ implies there is a $t_2'$ such that
$t_2\overto{\langle(\mathcal{C},\eta)\rangle}t_2'$ and $t_1'\simeq_{br} t_2'$;
\item $t_2 \overto{(\mathcal{C},\eta)}t_2'$ implies there is a $t_1'$ such that
$t_1\overto{\langle(\mathcal{C},\eta)\rangle} t_1'$ and $t_1'\simeq_{br} t_2'$.
\end{itemize}
\end{definition}

\begin{corollary}\label{Cor::eqiv}
Rooted branching reliable computed network bisimilarity is an equivalence.
\end{corollary}

Corollary \ref{Cor::eqiv} is an immediate consequence of Theorem~\ref{The::eqiv} and
Definition~\ref{Def::rbrbism}.

\begin{theorem}\label{The::cong}
Rooted branching reliable computed network bisimilarity is a congruence for
{\it RRBPT} operators.
\end{theorem}

See Section \ref{sec::cong} for the proof.

\section{Axiomatization for {\it RRBPT}}\label{sec::raxiom}
To provide a sound and complete axiomatization for closed {\it RRBPT} terms with respect to rooted branching reliable computed network bisimilarity, the framework should be extended with the computed network terms, i.e., $(\mathcal{C},\eta).t$ which expresses that action $\eta$ is possible for topologies belonging to $\mathcal{C}$, in the same way as \cite{FatemehFI10}. This prefix operator is helpful to transform protocol send/receive actions into their corresponding network ones. 
Furthermore, it borrows the operators \emph{left merge}
($\lm$) and \emph{communication merge} from the process algebra
\emph{ACP} \cite{BergstraKlop84}  to axiomatize
parallel composition. Note that the interleaving semantics for parallel composition is only valid for internal and unobservable actions (see rule $\it Par$). 
To axiomatize the behavior of nodes while being input-enabled, we also exploit two novel auxiliary operators.

{\it RRBPT} is extended with new operators and called \emph{Reliable Computed Network Process Theory} ({\it RCNT}). Its syntax contains:\[
\begin{array}{l}  t:: = ~  0 ~~ |~~ \beta.t  ~~ |~~ t+t~~|~~\mathfrak{A}~ ,\mathfrak{A}\deff t~~| ~~~
t\mid t  ~~|~~t\lm t ~~|~~t \parallel t~~|~~\rec\mathfrak{A}\cdot t\\
\hspace*{1cm}{\it sense}(\ell,t,t)~~|~~(\nu\ell)t~~|~~\tau_m(t)~~|~~ \partial_m(t)~~|~~\ell:t:t~~|~~\mathcal{C}\rhd t~~|~~ \oc  t\cc_\ell~~    \end{array}  \]The prefix operator in
$\beta.t$ again denotes a process which performs $\beta$ and then behaves as $t$. The action $\beta$ can now be of two types:
either an internal action or a send/receive action $\snd(\mathfrak{m})$/$\rcv(\mathfrak{m})$, denoted by $\alpha$, or actions of the form
$(\mathcal{C},\nrcv(\mathfrak{m}))$,    $(\mathcal{C},\nsnd(\mathfrak{m},\ell))$ and $(\mathcal{C},\tau)$, denoted by $(\mathcal{C},\eta)$, where the first two
actions are called the network receive and send actions, respectively. The new operator $\ell:t_1:t_2$, so-called \textit{local deployment}, defines the behavior of process $t_2$ deployed at the network address $\ell$ while it only considers the input-enabledness feature with regard to the behavior of $t_2$. In cases that it should drop a message (i.e., processing the message has not been defined by $t_2$), it behaves as $t_1$. This operator is helpful to axiomatize the behavior of the deployment operator in the reliable setting.  To axiomatize the behavior of the $\it sense$ operator, the framework is extended with the \emph{topology restriction} operator $\mathcal{C}\rhd t$ which restricts the behavior of $t$ by taking restrictions of $\mathcal{C}$ into account.

Due to the input-enabledness feature of nodes, their behavior is recursive: upon receiving a message for which no receive action has been defined, a node drops the message. To this aim, we exploit the recursion operator $\rec \mathfrak{A}\cdot t$, which specifies the \textit{solution} of the process name $\mathfrak{A}$, defined by the equation $\mathfrak{A} \deff t$. The process term $t_\mathfrak{A}  $ is a solution of the equation $\mathfrak{A} \deff t$ if the replacement of $\mathfrak{A}$ by $t_\mathfrak{A}$ on both sides of the equation results in equal terms, i.e. $t_\mathfrak{A}\simeq_{rb} t [t_\mathfrak{A }/\mathfrak{A}]$. As we are interested in equations with exactly one solution, we define a guardedness criterion for network names, in the same way as \cite{FatemehFI10}. A free
occurrence of a network name $A$ in $t$ is called
\textit{guarded} if this occurrence is in the scope of an action
prefix operator (not $(\mathcal{C},\tau)$ prefix) and not in the scope of an
abstraction operator~\cite{BaetenBravetti}; in other words, there is
a subterm $(\mathcal{C},\eta).t'$ in $t$ such that $\eta\neq\tau$, and
$\mathfrak{A}$ occurs in $t'$. $\mathfrak{A}$ is
\textit{(un)guarded} in $t$ if (not) every free occurrence of
$\mathfrak{A}$ in $t$ is guarded. A {\it RCNT} term $t$ is
\textit{guarded} if for every subterm $\rec \mathfrak{A}\cdot t'$,
$\mathfrak{A}$ is guarded in $t'$. This guardedness criterion
ensures that any guarded recursive term has a unique solution. 


A term is grammatically well-defined if its processes
deployed at a network address through either a network or local deployment operator, are only defined by action prefix,
choice, sense, and process names.

The operational semantic rules of the new operators are given in Table \ref{Tab::SOSRCNT} while the counterpart of ${\it Sync}_2$ holds. In these rules,
$t\noverto{\rcv(\mathfrak{m})}$ denotes that there
exists no $t'$ such that
$t\overto{(\mathcal{C}',~\rcv(\mathfrak{m}))}t'$ for some network constraint $\mathcal{C}'$.  
The behavior of the local deployment operator is almost similar to the deployment operator. Its rules ${\it Inter}_1'$ and ${\it Inter}_2'$ are the same as $\it Snd$ and ${\it Rcv}_1$, respectively. However, it substitutes ${\it Inter}_3'$ for ${\it Rcv}_2$ by which it only adds transitions containing the disconnectivity pair $?\nconn \ell$ for those possible receive actions of $t_2$ (generated by ${\it Rcv}_1$).
Rules ${\it Sen}_{3,4}$ make the behavior of ${\it sense}(\ell',t_1,t_2)$ input-enabled toward receive actions that are possible by $t_1$ but not $t_2$ and vice versa. The constraints of the topology restriction operator $\mathcal{C}\rhd t$ is added to the behaviors of $t$ as explained by the rule $\it TR$.

\begin{table*}[htbp]\caption{Semantics of the new operators of {\it RCNT} }\label{Tab::SOSRCNT}
\centering\begin{tabular}{c}
$\Rule{t_2\overto{(\mathcal{C},~\snd(\mathfrak{m}))}t_2'}{\ell:t_1:t_2  \overto{(\mathcal{C}[\ell/?],~\nsnd(\mathfrak{m},\ell))} \oc    t_2'\cc_\ell}:{\it Inter}_1'$~~~~$\Rule{t[\rec\mathfrak{A}\cdot t/\mathfrak{A}]\overto{(\mathcal{C},\eta)}t'}{\rec\mathfrak{A}\cdot t \overto{(\mathcal{C},\eta)}t'}:{\it Rec}$\vspace{4mm}\cr$\Rule{t_2\overto{(\mathcal{C},~\rcv(\mathfrak{m}))}t_2'}{\ell:t_1:t_2            \overto{(\mathcal{C}[\ell/?]\cup\{?\,\conn\,\ell\},~\nrcv(\mathfrak{m}))} \oc t_2'\cc_\ell}:{\it    Inter}_2'$\vspace{4mm}\cr $\Rule{t_2\overto{(\mathcal{C},~\rcv(\mathfrak{m}))}t_2'}{\ell:t_1:t_2\overto{(\mathcal{C}[\ell/?]\cup\{?\,\nconn\,\ell\},~\nrcv(\mathfrak{m}))}t_1}:{\it Inter}_3'$~~~$\Rule{t\overto{(\mathcal{C}',\eta)}
t'}{\mathcal{C}\rhd t\overto{(\mathcal{C}'\cup\mathcal{C} ,\eta)} t'}:{\it TR}$\vspace{4mm}
\cr
$\Rule{
t_1\noverto{\rcv(\mathfrak{m})}~~~t_2\overto{(\mathcal{C},~\rcv(\mathfrak{m}))}t_2' }{\ell:t_3: {\it
sense}(\ell',t_1,t_2)\overto{(\{?\,\conn\,\ell'\}\cup\mathcal{C},~\nrcv(\mathfrak{m}))}
t_3}:{\it Sen}_3$\vspace{4mm}\cr
$\Rule{
t_1\overto{(\mathcal{C},~\rcv(\mathfrak{m}))}t_1'~~~t_2\noverto{\rcv(\mathfrak{m})} }{\ell:t_3: {\it
sense}(\ell',t_1,t_2)\overto{(\{?\,\nconn\,\ell'\}\cup\mathcal{C},~\nrcv(\mathfrak{m}))}
t_3}:{\it Sen}_4$~~~$\Rule{t_1\overto{\beta}
    t_1'}{t_1\lm t_2\overto{\beta}
    t_1'\parallel t_2}:~{\it LExe}$\vspace{4mm}\cr
$\Rule{t_1\overto{(\mathcal{C}_1,\nrcv(\mathfrak{m}))}
t_1'~~~t_2\overto{(\mathcal{C}_2,\nrcv(\mathfrak{m}))}
t_2'}{t_1\mid t_2
\overto{(\mathcal{C}_1\cup \mathcal{C}_2,\nrcv(\mathfrak{m}))} t_1'\parallel
t_2'}:{\it Sync}_1$ 
\vspace{4mm}\cr
$\Rule{t_1\overto{(\mathcal{C}_1,\nsnd(\mathfrak{m},\ell))}
t_1' ~~ t_{2} \overto{(\mathcal{C}_2,\nrcv(\mathfrak{m}))}
t_{2}'}{t_1 \mid
t_2\overto{(\mathcal{C}_1\cup \mathcal{C}_2[\ell/?],\nsnd(\mathfrak{m},\ell))}
t_1'\parallel t_2'}:{\it
Sync}_2$
\end{tabular}
\end{table*}

The main differences of extended {\it RCNT} with {\it CNT} are that its deployed nodes are input-enabled and its communication primitive is reliable. We use the notation $\sum_{m\in M} t$ to define $t[m_1/\mathfrak{m}]+\ldots+t[m_k/\mathfrak{m}]$, where $M=\{m_1,\ldots,m_k\}$. Furthermore, ${\it if}(b,t_1,t_2)$ behaves as $t_1$ if the condition $b$ holds and otherwise as $t_2$.

The axioms regarding the choice, deployment, left and communication merge, and parallel operators are given in Table \ref{Tab::reliableparallelaxiom}. The axioms ${\it Ch}_{1-4}$, $\it Br$, ${\it LM}_{2,3}$ and $S_{1-4}$ are
standard (cf.\ \cite{thesis-Luttik}). The axiom ${\it Ch}_5$
denotes that a network send action
whose sender address is unknown can be
removed if its counterpart action exists. The axiom ${\it Ch}_6$ explains that a more liberal network constraint allows more behavior. Axioms ${\it Dep}_{0-7}$, ${\it LM}_{1,2}'$, and ${\it TRes}_{1-5}$ are new in comparison with the lossy setting of \cite{FatemehFI10}. The axiom $(\mathcal{C},\eta).t_1\lm t_2=(\mathcal{C},\eta).(t_1\parallel t_2)$ has been replaced by ${\it LM}_{1,2}'$ which only allow internal or unobservable actions of the left operand to be performed.

To axiomatize the behavior of a node considering the input-enabledness feature, we need to find the messages that it cannot currently respond to and then add a summand which receives those message without processing them. To this aim, axiom ${\it Dep}_0$ expresses the behavior of $
\oc t\cc_\ell$ as a recursive specification which drops messages that it does not handle with the help of the auxiliary function ${\it Message}(t,\mathcal{S})$, and the behavior of $t$ with the help of the local deployment operator $\ell: \mathfrak{Q}: t$. The function ${\it Message}(t,\mathcal{S})$ returns the set of messages that can be currently processed by $t$ and is defined using structural induction:\[\begin{array}{l}
{\it Message}(0,\mathcal{S})=\emptyset\\
{\it Message}(i.t,\mathcal{S})=\emptyset,~i\in{\it IAct}\\
{\it Message}(\snd(\mathfrak{m}).t,\mathcal{S}) = \emptyset\\
{\it Message}(\rcv(\mathfrak{m}).t,\mathcal{S}) = \{\mathfrak{m}\}\\
{\it Message}(t_1+t_2,\mathcal{S})={\it Message}(t_1,\mathcal{S})\cup{\it Message}(t_2,\mathcal{S})\\
{\it Message}({\it sense}(\ell,t_1,t_2),\mathcal{S})={\it Message}(t_1,\mathcal{S})\cup{\it Message}(t_2,\mathcal{S})\\
{\it Message}(\mathfrak{A},\mathcal{S})={\it Message}(t,\mathcal{S}\cup\{\mathfrak{A}\}),~ \mathfrak{A}\not\in\mathcal{S},\mathfrak{A}\deff t\\
{\it Message}(\mathfrak{A},\mathcal{S})=\emptyset,~ \mathfrak{A}\in\mathcal{S}
\end{array}
\]where $\mathcal{S}$ keeps track of process names whose right-hand definitions have been examined. We remark that ${\it Dep}_0$ extends the deployment behavior of the lossy setting with the input enabledness feature with the help of operator $\ell:\mathfrak{Q}:t$. The axioms ${\it Dep}_{1-7}$ specify the behavior of the operator $\ell:t_1:t_2$. 
Axiom ${\it Dep}_1$ defines the interaction between the network and data link layers. The protocol send action (at the network layer) is transformed into its network version (at the data link layer). Axiom ${\it Dep}_2$ indicates that when $\ell$ is connected to a sender (which is unknown yet), the receive action is successful and its behavior proceeds as $\oc t\cc_\ell$. Otherwise, the receive action is unsuccessful and its behavior is defined by $t'$. Axioms ${\it Dep}_{3,4,5}$ express the effect of the local deployment on choice, deadlock, and process names, respectively while axioms ${\it Dep}_{6,7}$ define its effect on the prefixed internal actions and ${\it sense}$ operator, respectively.

\begin{table}\small
\centering
\caption{Axioms for the choice, deployment, left and communication merge, and parallel operators. The sets $M_1$ and $M_2$ denote ${\it Message}(t_2,\emptyset)\setminus {\it Message}(t_1,\emptyset)$ and ${\it Message}(t_1,\emptyset)\setminus {\it Message}(t_2,\emptyset)$  respectively.}\label{Tab::reliableparallelaxiom}
\begin{tabular}{|l|}
\hline
\begin{tabular}{llll}
${\it Ch}_1$ & $0+t=t$ & ${\it Ch}_2$ & $t_1+t_2=t_2+t_1$ \cr ${\it Ch}_3$ &
$t_1+(t_2+t_3)=(t_1+t_2)+t_3$ & ${\it Ch}_4$ & $t+t=t$\cr ${\it Ch}_{5}$ & \multicolumn{3}{l}{$(\mathcal{C},\nsnd(\mathfrak{m},?)).t+\langle(\mathcal{C},\nsnd(\mathfrak{m},?))\rangle.t=\langle(\mathcal{C},\nsnd(\mathfrak{m},?))\rangle.t$}\cr ${\it Ch}_{6}$ &
\multicolumn{3}{l}{$(\mathcal{C}_1,\eta).t+(\mathcal{C}_2,\eta).t=(\mathcal{C}_1,\eta).t,~\mathcal{C}_2\less
\mathcal{C}_1$} \vspace*{2mm}\cr
${\it Dep}_0$ & \multicolumn{3}{l}{$\oc t\cc_\ell = \rec\mathfrak{Q}\cdot\sum_{\mathfrak{m}'\not\in{\it Message}(t,\emptyset)}(\{\},\nrcv(\mathfrak{m}')).\mathfrak{Q}+\ell:\mathfrak{Q}:t$} \cr
${\it Dep}_2$ & \multicolumn{3}{l}{$\ell:t': \rcv(\mathfrak{m}).t =(\{?\nconn\ell\},\nrcv(\mathfrak{m})).t'+(\{?\conn\ell\},\nrcv(\mathfrak{m})).\oc t\cc_\ell$}\\
${\it Dep}_7$ & \multicolumn{3}{l}{$\ell:t_3:{\it sense}(\ell',t_1,t_2) = \sum_{\mathfrak{m}'\in M_1}(\{\ell\conn \ell'\},\nrcv(\mathfrak{m}')).t_3$}\\
& \multicolumn{3}{l}{$ + \sum_{\mathfrak{m}'\in M_2}(\{\ell\nconn \ell'\},\nrcv(\mathfrak{m}')).t_3+\{\ell\conn \ell'\}\rhd \ell:t_3: t_1 + \{\ell\nconn \ell'\}\rhd \ell:t_3: t_2$}\vspace*{2mm}\cr
${\it Dep}_1$ & $\ell:t':\snd(\mathfrak{m}).t= (\{\},\nsnd(\mathfrak{m},\ell)).\oc t\cc_\ell$& ${\it Dep}_6$ & $\ell:t': i.t = (\{\},i).\oc t\cc_\ell$  \\
${\it Dep}_3$ & $\ell:t_3: t_1+t_2=\ell: t_3:t_1+\ell:
t_3:t_2$ & ${\it Dep}_4$ &$\ell: t: 0 = 0$\cr
${\it Dep}_5$ & $\ell:t':\mathfrak{A}= \ell:t': t,~\mathfrak{A}\deff t$ & & 
\vspace*{2mm}\cr
${\it TRes}_1$ & \multicolumn{3}{l}{$\mathcal{C}_1\rhd (\mathcal{C}_2,\eta).t = (\mathcal{C}_1\cup \mathcal{C}_2,\eta).t$, ~if $\mathcal{C}_1\cup \mathcal{C}_2 \in \mathbb{C}^v({\it Loc})$}    \cr
${\it TRes}_2$ & $\mathcal{C}\rhd (t_1+t_2) = (\mathcal{C}\rhd t_1) + (\mathcal{C}\rhd t_2)$ & ${\it TRes}_3$ & $\mathcal{C}\rhd\rec\mathfrak{A}\cdot t = \rec\mathfrak{A}\cdot(\mathcal{C}\rhd t)$ \cr ${\it TRes}_4$ & $\mathcal{C}\rhd \mathfrak{A} = \mathfrak{A}$&  ${\it TRes}_5$ & $\mathcal{C}\rhd 0 = 0$\vspace*{2mm} \cr

\end{tabular}\cr
\begin{tabular}{llll}
${\it Br}$ & $t_1\parallel t_2=t_1\lm t_2+t_2\lm t_1+t_1\mid t_2$ &
$S_1$ & $t_1\mid t_2=t_2\mid t_1$\cr
${\it LM}_1'$ & $(\mathcal{C},\eta).t_1\lm t_2 = 0, ~\eta\not\in {\it IAct}\cup \{\tau\}$ & $S_2$ & $(t_1+t_2)\mid t_3=t_1\mid t_3+t_2\mid t_3$\cr
${\it LM}_2$ & $(t_1+t_2)\lm t_3=t_1\lm t_3+t_2\lm t_3$ & $S_3$ & $0\mid t = 0$\cr
${\it LM}_3$ & $0\lm t = 0$ & $S_4$ & $(\mathcal{C},\eta).t_1\mid
t_2=0,~\eta\in {\it IAct}\cup \{\tau\}$\cr 
${\it LM}_2'$ & \multicolumn{3}{l}{$(\mathcal{C},\eta).t_1\lm t_2 = (\mathcal{C},\eta).(t_1\parallel t_2), ~\eta\in {\it IAct}\cup \{\tau\}$}  \vspace*{2mm}\cr
\end{tabular}\cr
${\it Sync}_1$~~$(\mathcal{C}_1,\nsnd(\mathfrak{m}_1,\ell)).t_1\mid
(\mathcal{C}_2,\nrcv(\mathfrak{m}_2)).t_2=$\cr
\hspace*{2cm}${\it if}(
(\mathfrak{m}_1=\mathfrak{m}_2),(\mathcal{C}_1\cup\mathcal{C}_2[\ell/?],\nsnd(\mathfrak{m}_1,\ell)).t_1\parallel
t_2, 0)$\cr
${\it Sync}_2$~~$(\mathcal{C}_1,\nrcv(\mathfrak{m}_1)).t_1\mid
(\mathcal{C}_2,\nrcv(\mathfrak{m}_2)).t_2={\it if}(
        (\mathfrak{m}_1=\mathfrak{m}_2),(\mathcal{C}_1\cup\mathcal{C}_2,\nrcv(\mathfrak{m}_1)).t_1\parallel
t_2, 0)$\cr
${\it Sync}_3$~~$(\mathcal{C}_1,\nsnd(\mathfrak{m}_1,\ell_1)).t_1\mid
(\mathcal{C}_2,\nsnd(\mathfrak{m}_2,\ell_2)).t_2=0$\cr
\hline
\end{tabular}
\end{table}

The behavior of the topology restriction operator is defined by the axioms ${\it TRes}_{1-5}$ in Table \ref{Tab::reliableparallelaxiom}. Axiom ${\it TRes}_1$ considers the restrictions of $\mathcal{C}_1$ by integrating its restrictions with $\mathcal{C}_2$ in the computed network term $(\mathcal{C}_2,\eta).t$ if $\mathcal{C}_1\cup \mathcal{C}_2$ is well-formed. Axiom ${\it TRes}_{2}$ defines that topology restriction can be distributed over the choice operator. Axiom ${\it TRes}_{3}$ expresses that the topology restriction operator can be moved inside and outside of a recursion operator. Axioms ${\it TRes}_{4,5}$ explain that the topology restriction operator has no effect on a process name and deadlock, respectively.

For instance, the behavior of the MANET $\oc P\cc_A$, where $P\deff {\it sense}(B, \snd(\data_B).P,0)$, ${\it Msg}=\{\data_B\}$, is simplified as: {{\allowdisplaybreaks
\begin{flalign*}
& \oc P\cc_A =^{{\it Dep}_{0,5}}\\ &
\hspace*{0.cm}\rec\mathfrak{Q}\cdot (\{\},\nrcv(\data_B)).\mathfrak{Q} + A:\mathfrak{Q}:{\it sense}(B, \snd(\data_B).P,
0) =^{{\it Dep}_7}  \\ &
\hspace*{0.cm} \rec\mathfrak{Q}\cdot (\{\},\nrcv(\data_B)).\mathfrak{Q}+ 
\{A\conn B\}\rhd A:\mathfrak{Q}:\snd(\data_B).P  +\{A\nconn B\}\rhd A:\mathfrak{Q}: 0 =^{{\it Dep}_{1,4}}\\
&\hspace*{0.cm}\rec\mathfrak{Q}\cdot (\{\},\nrcv(\data_B)).\mathfrak{Q}+ \{A\conn B\}\rhd (\{\},\nsnd(\data_B,A)).\mathfrak{Q}+ 
\{B\nconn A\}\rhd 0 =^{{\it TRes}_{1,5}}\\ & \hspace*{0.cm}\rec\mathfrak{Q}\cdot (\{\},\nrcv(\data_B)).\mathfrak{Q}+ (\{A\conn B\},\nsnd(\data_B,A)).\mathfrak{Q}
\end{flalign*}}%
}
The behavior of $\oc Q\cc_B$, where $Q\deff \rcv(\data_B).\deliver.Q$, is equated to:

{{\allowdisplaybreaks
\begin{flalign*}
\qquad\oc & Q\cc_B =^{{\it Dep}_{0,5}}\\ &
\hspace*{0.5cm}\rec\mathfrak{Q}\cdot  A:\mathfrak{Q}:\rcv(\data_B).\deliver.Q =^{{\it Dep}_2}  \\ &
\hspace*{0.5cm} \rec\mathfrak{Q}\cdot (\{?\conn B\},\nrcv(\data_B)).\oc \deliver.Q\cc_A + (\{?\nconn B\},\nrcv(\data_B)).\mathfrak{Q}
\end{flalign*}}%
}

\begin{table}\small
\centering
\caption{Axiomatization of hiding, abstraction and encapsulation operators.}\label{Tab::abstaxiom}
\begin{tabular}{|llll|}
\hline  ${\it Res}_1$ & $(\nu\ell)(t_1+t_2)=(\nu\ell) t_1+(\nu\ell)
t_2$ & ${\it Res}_3$ & $(\nu\ell)0=0$ \cr ${\it Res}_2$ &
$(\nu\ell)(\mathcal{C},\eta).t=
(\mathcal{C}[?/\ell],\eta[?/\ell]).(\nu\ell)t$ & & 
\cr & & & \cr
${\it Ecp}_1$ &
\multicolumn{3}{l|}{$\partial_m((\mathcal{C},\nsnd(\mathfrak{m},\ell)).t)=(\mathcal{C},\nsnd(\mathfrak{m},\ell)).\partial_m(t)$}
\cr ${\it Ecp}_2$ &
\multicolumn{3}{l|}{$\partial_m((\mathcal{C},\nrcv(\mathfrak{m})).t)={\it if}((m\neq \mathfrak{m}),(\mathcal{C},\nrcv(\mathfrak{m})).\partial_m(t),
0)$} \cr & & & \cr
${\it Abs}_1$ &
\multicolumn{3}{l|}{$\tau_m((\mathcal{C},\nrcv(\mathfrak{m})).t)={\it if}((m= \mathfrak{m}),(\mathcal{C},\tau).\tau_m(t), (\mathcal{C},\nrcv({\mathfrak{m}})).\tau_m(t))$} \cr
${\it Abs}_2$ &
\multicolumn{3}{l|}{$\tau_m((\mathcal{C},\nsnd(\mathfrak{m},\ell)).t)={\it if}((m= \mathfrak{m}),(\mathcal{C},\tau).\tau_m(t), (\mathcal{C},\nsnd(\mathfrak{m},\ell)).\tau_m(t))$} \cr
${\it Abs}_3$ & $\tau_m(t_1+t_2)=\tau_m(t_1)+\tau_m(t_2)$
& ${\it Ecp}_3$ &
$\partial_m(t_1+t_2)=\partial_m(t_1)+\partial_m(t_2)$ \cr ${\it
Abs}_4$ & $\tau_m(0)=0$ & ${\it Ecp}_4$ & $\partial_m(0)=0$  \cr
& & & \cr
${\it T}_1$ &
\multicolumn{3}{l|}{$(\mathcal{C}',\eta).((\mathcal{C}_1,\eta).t+(\mathcal{C}_2,\eta).t+t')=(\mathcal{C}',\eta).((\mathcal{C},\eta).t+t')$}\cr
& \multicolumn{3}{l|}{$\mbox{iff } \exists \ell,\ell'\in{\it Loc},\exists \mathcal{C}\in\mathbb{C}^v({\it Loc})\cdot ( \mathcal{C}_1=\mathcal{C}\cup\{\ell\conn\ell' \}\,\wedge\,\mathcal{C}_2=\mathcal{C}\cup\{\ell\nconn \ell'\}$}
\cr ${\it T}_2$ &
\multicolumn{3}{l|}{$(\mathcal{C},\eta).((\mathcal{C}',\tau).(t_1+t_2)+t_2)=(\mathcal{C},\eta).(t_1+t_2)$} \cr
\hline
\end{tabular}

\end{table}

The axioms of hiding and encapsulation are given in Table
\ref{Tab::abstaxiom}. Axiom $T_1$ accumulates the network constraints that constitute a partitioning while $T_2$ removes a $\tau$ action which preserves the behavior of a network after some
topology changes. The remaining axioms in this table are  similar to the lossy setting.

Axioms for process names are given in Table~\ref{Tab::recaxiom}. $\it Unfold$ and $\it Fold$ express existence and uniqueness of a
solution for the equation $\mathfrak{A} \deff t$, which correspond to
Milner's standard axioms, and the \textit{Recursive Definition
Principle} ({\it RDP}) and \textit{Recursive Specification
Principle} ({\it RSP}) in {\it ACP}. ${\it Unfold}$ states that each
recursive operator has a solution (whether it is guarded or not),
while ${\it Fold}$ states that each guarded recursive operator has
at most one solution.

The behavior of $\tau_{\it Msg}(\partial_{\it Msg}(\oc P\cc_A\parallel \oc Q\cc_B))$ by using the axioms of Table ~\ref{Tab::recaxiom} is expressed by:
{{\allowdisplaybreaks
\begin{flalign*}
\qquad\tau_{\it Msg}(\partial_{\it Msg}(\oc & P\cc_B \parallel\oc Q\cc_B)) =\\ &
\hspace*{0.5cm}\rec\mathfrak{Q}\cdot  (\{A\conn B\},\tau).(\{\},\deliver).\mathfrak{Q}+ (\{A\nconn B\},\tau).0
\end{flalign*}}%
}

\noindent
which explains that in case $A$ is connected to $B$, each sending of $\data_B$ is followed by the internal action $\it deliver$

\begin{table}\small
\centering
\caption{Axioms for process names.}\label{Tab::recaxiom}    \begin{tabular}{|ll|}
\hline
$\rec \mathfrak{A}\cdot t = t\{\rec \mathfrak{A} \cdot t/ \mathfrak{A}\}$ &
${\it Unfold}$\cr
$t_1=t_2\{t_1/\mathfrak{A}\}\Rightarrow t_1=\rec \mathfrak{A}\cdot t_2, \mbox{\small ~if $A$ is guarded in $t_2$}$ & ${\it Fold}$\cr 
$\rec \mathfrak{A}\cdot (\mathfrak{A}+t)=\rec \mathfrak{A}\cdot t$ & ${\it
Ung}$\cr $\rec \mathfrak{A}\cdot  ((\mathcal{C},\tau). ((\mathcal{C}',\tau). t'+t)+s) =$ & ${\it
WUng}_1$\cr $~~~~~~~~\rec
\mathfrak{A}\cdot ((\mathcal{C},\tau). (t'+t)+s),\mbox{\small~ if $\mathfrak{A}$ is unguarded in $t'$}
$ &  \cr $\rec \mathfrak{A}\cdot
((\mathcal{C},\tau).(\mathfrak{A}+t)+s)=\rec \mathfrak{A}\cdot
((\mathcal{C},\tau). (t+s)+s)$ & ${\it
WUng}_2$\cr
$\tau_m(\rec
\mathfrak{A}\cdot t)=\rec \mathfrak{A}\cdot\tau_m(t),\mbox{\small
~ if $\mathfrak{A}$ is serial in $t$}$ & ${\it Hid}$\cr
\hline
\end{tabular}
\end{table}


It is not hard to see that the axioms of Table \ref{Tab::reliableparallelaxiom}, Table \ref{Tab::abstaxiom} and Table \ref{Tab::recaxiom} provide a sound axiomatization of {\it RCNT}. This can be checked by verifying soundness for each axiom individually. 

\begin{theorem}\label{The::soundnessCNT}
The axiomatization is sound,
i.e.\ for all closed {\it RCNT} terms
$t_1$ and $t_2$, if
$t_1=t_2$ then $t_1\simeq_{rb}
t_2$.
\end{theorem}

Our axiomatization is also ground-complete for terms with a finite-state CLTS, but not for infinite-state CLTSs. For example,
$\rec \mathfrak{W}\cdot (\{\},\nsnd({\it req}(A),A)).  \mathfrak{W}\parallel \sum_{lx:{\it Loc}}(\{?\conn B\},\nrcv({\it req}(lx))). \mathfrak{W} $ produces
an infinite-state CLTS, since at each recursive call a
new parallel operator is generated. Its equality to $\rec \mathfrak{H}\cdot (\{\},\nsnd({\it req}(A),A)). \mathfrak{H}$ cannot be proved by our axiomatization.

\begin{theorem}\label{The::completenessCNT}
The axiomatization is ground-complete, i.e., for all closed finite-state reliable computed
network terms $t_1$ and $t_2$,
$t_1\simeq_{rb}t_2$ implies $t_1=
t_2$.
\end{theorem}

See sections \ref{sec::soundness} and \ref{sec::completeness} for the proofs of theorems \ref{The::soundnessCNT} and \ref{The::completenessCNT}, respectively.
\section{Case Study}\label{sec::casestudy}
In MANETs, nodes communicate through others via a multi-hop communication. Hence, nodes act as routers to make the communication possible among not directly connected nodes. We illustrate the applicability of our axioms in the analysis of MANET protocols through a simple routing protocol inspired by the AODV protocol.

\subsection{Protocol Specification}
The protocol consists of three processes $P$, $M$, and $Q$, each specifying the behavior of a node as the source (that finds a route to a specific destination), middle node (that relays messages from the source to the destination), and destination. The description of these process are given in Figure \ref{Fig::protocol}.

\begin{figure} [hptb]
\begin{eqnarray*}
P &\deff& \sense(B,\snd(\data_B).P,\snd(\req).P_1)\\
P_1 &\deff& [\rcv(\rep_C).P_2+\rcv(\rep_B).P+\snd(\req).P_1]\\\
P_2 &\deff& \sense(C,\rcv(error).P + \snd(\data_C).P_2,\snd(\req).P_1)\\
M &\deff& \rcv(\req).\snd(\req).M_1\\
M_1 &\deff& \rcv(\rep_B).\snd(\rep_C).M_2+\snd(\req).M_1 \\
M_2 &\deff& \sense(B, \rcv(\data_C).\snd(\data_B).M_2,\snd(error).\snd(\req).M_1)\\
Q&\deff& \rcv(\req).\snd(\rep_B).Q+\rcv(\data_B).deliver.Q
\end{eqnarray*}
\caption{The specification of processes $P$, $M$, and $Q$ as a part of our simple routing protocol.}\label{Fig::protocol}
\end{figure}

Process $P$, deployed at the address $A$, uses the neighbor discovery service of the data link layer to examine if it has a direct link to the destination with the address $B$. If it is connected, then it sends its data directly by broadcasting the message $\data_B$; otherwise, it initiates the route discovery procedure by sending the message $\req$, then behaving as $P_1$. This process waits until it receives a reply from a middle name with the address $C$ or $B$. In the former case, it behaves as $P_2$ which indicates that $A$ sends it data through $C$ as long as $C$ is connected to $A$. In the latter case, it behaves as $P$ which indicates that $A$ sends it data as long as $B$ is directly connected to $A$.

Process $M$ relays $\req$ messages to find a route to $B$ and then behaves as $M_1$. This process waits until it receives a reply. To model waits with a timeout, it non-deterministically sends a request again. Upon receiving a reply from $C$ it behaves as $M_2$, indicating that it relays data messages of $A$ as long as it has a link to $C$. Finally, process $Q$ sends a reply upon receiving a request message and receives data messages.

To simplify the route maintenance procedure of AODV, the middle node takes advantage of the sensing operator when it behaves as $M_2$. Whenever it finds out that it has no link to $C$, it sends an error message to its upstream node, i.e., $A$, to inform it that its route to $B$ through $C$ is not valid. Afterwards, they both execute the route discovery procedure by sending a request message.

The network with the three nodes of a source, middle, and destination is specified by \[\mathcal{N} \equiv \tau_\Msg(\partial_\Msg(\oc P\cc_A\parallel \oc M\cc_C\parallel \oc Q\cc_B)).\] Analyzing $(\nu A)(\nu B)(\nu C)\mathcal{N}$, whose network addresses have been abstracted away, reveals that it is rooted branching bisimilar to $\rec X\cdot \tau.\deliver.X+\tau.0$. 
Thus, possibly a deadlock occurs where data is not delivered to $B$. Such behavior may be the result of a conceptual mistake in the protocol design or lossy communication between $A$ and $B$. However, the latter one does not exist in our reliable setting. 
We propose a technique in Section \ref{sec::precon} to discover only those faulty behaviors that are due to an incorrect protocol design. 

The network $\partial_\Msg(\oc P\cc_A\parallel \oc M\cc_C\parallel \oc Q\cc_B)$ can be simplified as:
{{\allowdisplaybreaks
\begin{flalign}\label{eq::base}
\qquad\partial_\Msg(\oc& P\cc_A\parallel \oc M\cc_C\parallel \oc Q\cc_B)=\\ &
(\{A\conn B\},\nsnd(\data_B,A)).\partial_\Msg(\oc P\cc_A\parallel \oc M\cc_C\parallel \oc \deliver.Q\cc_B)+\nonumber\\
&(\{A\nconn B,A\conn C\},\nsnd(\req,A)).\partial_\Msg(\oc P_1\cc_A\parallel \oc\snd(\req).M_1\cc_C\parallel \oc Q\cc_B)+\nonumber\\
&(\{A\nconn B,A\nconn C\},\nsnd(\req,A)).\partial_\Msg(\oc P_1\cc_A\parallel \oc M\cc_C\parallel \oc Q\cc_B).\nonumber
\end{flalign}}%
}
Next, we simplify $\partial_\Msg(\oc P_1\cc_A\parallel \oc\snd(\req).M_1\cc_C\parallel \oc Q\cc_B)$ as
{{\allowdisplaybreaks
\begin{flalign}\label{eq::m}
\qquad\partial_\Msg(\oc&  P_1\cc_A\parallel \oc\snd(\req).M_1\cc_C\parallel \oc Q\cc_B)=\\ &
(\{A\conn B\},\nsnd(\req,A)).\partial_\Msg(\oc P_1\cc_A\parallel \oc \snd(\req).M_1\cc_C\parallel \oc \snd(\rep_B).Q\cc_B)+\nonumber\\
&(\{A\nconn B\},\nsnd(\req,A)).\partial_\Msg(\oc P_1\cc_A\parallel \oc \snd(\req).M_1\cc_C\parallel \oc Q\cc_B)+\nonumber\\
&(\{C\conn B\},\nsnd(\req,C)).\partial_\Msg(\oc P_1\cc_A\parallel \oc M_1\cc_C\parallel \oc \snd(\rep_B).Q\cc_B)+\nonumber\\
&(\{C\nconn B\},\nsnd(\req,C)).\partial_\Msg(\oc P_1\cc_A\parallel \oc M_1\cc_C\parallel \oc Q\cc_B).\nonumber
\end{flalign}}%
}
Now, we continue by extending $\partial_\Msg(\oc P_1\cc_A\parallel \oc M_1\cc_C\parallel \oc \snd(\rep_B).Q\cc_B)$:
{{\allowdisplaybreaks
\begin{flalign*}
\qquad\partial_\Msg(\oc&  P_1\cc_A\parallel \oc M_1\cc_C\parallel \oc \snd(\rep_B).Q\cc_B)=\\ &
(\{~\},\nsnd(\req,A)).\partial_\Msg(\oc P_1\cc_A\parallel \oc M_1\cc_C\parallel \oc \snd(\rep_B).Q\cc_B)+\\
&(\{~\},\nsnd(\req,C)).\partial_\Msg(\oc P_1\cc_A\parallel \oc M_1\cc_C\parallel \oc \snd(\rep_B).Q\cc_B)+\\
&(\{B\conn A,C\},\nsnd(\rep_B,B)).\partial_\Msg(\oc P\cc_A\parallel \oc \snd(\rep_C).M_2\cc_C\parallel \oc Q\cc_B)+\\
&(\{B\conn A,B\nconn C\},\nsnd(\rep_B,B)).\partial_\Msg(\oc P\cc_A\parallel \oc M_1\cc_C\parallel \oc Q\cc_B)+\\
&(\{B\nconn A,B\conn C\},\nsnd(\rep_B,B)).\partial_\Msg(\oc P_1\cc_A\parallel \oc\snd(\rep_C).M_2\cc_C\parallel \oc Q\cc_B).
\end{flalign*}}%
}By simplifying the term $\partial_\Msg(\oc P\cc_A\parallel \oc \snd(\rep_C).M_2\cc_C\parallel \oc Q\cc_B)$, which indicates that $A$ and $C$ have found a direct route to $B$, we reach $\partial_\Msg(\oc P\cc_A\parallel \oc M_2\cc_C\parallel \oc Q\cc_B)$:
{{\allowdisplaybreaks
\begin{flalign*}
\qquad\partial_\Msg(\oc&  P\cc_A\parallel \oc \snd(\rep_C).M_2\cc_C\parallel \oc Q\cc_B)=\\ &
(\{~\},\nsnd(\rep_C,C)).\partial_\Msg(\oc P\cc_A\parallel \oc M_2\cc_C\parallel \oc Q\cc_B)+\\
&(\{A\conn B\},\nsnd(\data_B,A)).\partial_\Msg(\oc P\cc_A\parallel \oc \snd(\rep_C).M_2\cc_C\parallel \oc \deliver.Q\cc_B)+\\
&(\{A\nconn B\},\nsnd(\req,A)).\partial_\Msg(\oc P_1\cc_A\parallel \oc\snd(\rep_C).M_2\cc_C\parallel \oc Q\cc_B).\\
\end{flalign*}}%
}By extending $\partial_\Msg(\oc P\cc_A\parallel \oc M_2\cc_C\parallel \oc Q\cc_B)$, we have:
{{\allowdisplaybreaks
\begin{flalign*}
\qquad\partial_\Msg(\oc&   P\cc_A\parallel \oc M_2\cc_C\parallel \oc Q\cc_B)=\\ &
(\{A\conn B\},\nsnd(\data_B,A)).\partial_\Msg(\oc P\cc_A\parallel \oc M_2\cc_C\parallel \oc \deliver.Q\cc_B)+\\
&(\{A\nconn B\},\nsnd(\req,A)).\partial_\Msg(\oc P_1\cc_A\parallel \oc M_2\cc_C\parallel \oc Q\cc_B).
\end{flalign*}}%
}
Finally extending $\partial_\Msg(\oc P_1\cc_A\parallel \oc M_2\cc_C\parallel \oc Q\cc_B)$ results:
{{\allowdisplaybreaks
\begin{flalign*}
\qquad\partial_\Msg(\oc&   P_1\cc_A\parallel \oc M_2\cc_C\parallel \oc Q\cc_B)=\\ &
(\{A\conn B\},\nsnd(\req,A)).\partial_\Msg(\oc P_1\cc_A\parallel \oc M_2\cc_C\parallel \oc \snd(\rep_B).Q\cc_B)+\\
&(\{A\nconn B\},\nsnd(\req,A)).\partial_\Msg(\oc P_1\cc_A\parallel \oc M_2\cc_C\parallel \oc Q\cc_B)+\\
&(\{C\nconn B\},\nsnd({\it error},C)).\partial_\Msg(\oc P_1\cc_A\parallel \oc \snd(\req).M_1\cc_C\parallel \oc Q\cc_B).
\end{flalign*}}%
}
The following scenario, found by above equations, is valid for a topology in which $A$ has only a multi-hop link to $B$ via $C$, but $B$ has a direct link to $A$: \[
\begin{array}{l}
\partial_\Msg(\oc P\cc_A\parallel \oc M\cc_C\parallel \oc Q\cc_B)\\
\hspace*{2cm}\overto{(\{A\nconn B,A\conn C\},\nsnd(\req,A))}\partial_\Msg(\oc P_1\cc_A\parallel \oc\snd(\req).M_1\cc_C\parallel \oc Q\cc_B)\\
\hspace*{2cm}\overto{(\{C\conn B\},\nsnd(\req,C))}\partial_\Msg(\oc P_1\cc_A\parallel \oc M_1\cc_C\parallel \oc \snd(\rep_B).Q\cc_B)\\
\hspace*{2cm} \overto{(\{B\conn A,C\},\nsnd(\rep_B,B))}\partial_\Msg(\oc P\cc_A\parallel \oc \snd(\rep_C).M_2\cc_C\parallel \oc Q\cc_B)\\\hspace*{2cm}\overto{(\{~\},\nsnd(\rep_C,C))}\partial_\Msg(\oc P\cc_A\parallel \oc M_2\cc_C\parallel \oc Q\cc_B)\\
\hspace*{2cm}\overto{(\{A\nconn B\},\nsnd(\req,A))}\partial_\Msg(\oc P_1\cc_A\parallel \oc M_2\cc_C\parallel \oc Q\cc_B)\\\hspace*{2cm}\overto{(\{A\nconn B\},\nsnd(\req,A))}\partial_\Msg(\oc P_1\cc_A\parallel \oc M_2\cc_C\parallel \oc Q\cc_B) \\\hspace*{2cm}\ldots
\end{array}
\]The reason is found in the specification of $M_2$ which does not handle request messages, and hence, for such a topology no data will be received by $B$ although there is a path form $A$ to $B$ and from $B$ to $A$. Therefore, we revise $M_2$ as:
\begin{eqnarray*}
M_2 &\deff &\sense(B, \rcv(\data_C).\snd(\data_B).M_2+\rcv(\req).\snd(\rep_C).M_2,\\ &&\hspace*{1.4cm}\snd(error).\snd(\req).M_1)
\end{eqnarray*}

The path above also exists in the lossy setting, but with all disconnectivity pairs removed from the network constraints. However, an exhaustive and therefore expensive inspection of this path is needed to determine that it is due to a design error. The first transition carries the label $(\{A\nconn B,A\conn C\},\nsnd(\req,A))$ in the reliable setting, meaning that $B$ is not ready to receive, and the label $(\{A\conn C\},\nsnd(\req,A))$ in the lossy setting. The latter label indicates that either $B$ was not ready to receive or it was not connected to $A$. So in the lossy setting one has to examine the origin state to find out if $B$ had an enabled receive action or not. The concept of not being ready to receive is treated in the same way as a lossy communication. Since only the former may be due to a conceptual design in the protocol, finding design errors is not straightforward in the lossy setting. In general the lossy setting will produce a large number of possible error traces that all need to be examined exhaustively, while the reliable setting will produce no spurious error traces.

\subsection{Protocol Analysis}\label{sec::precon}
The properties of wireless protocols, specially MANETs, tends to be weaker in comparison with wired protocols. For instance, the simple property of packet delivery from node $A$ to $B$ is specified as ``if there is a path from $A$ to $B$ for a long enough period of time, any packet sent by $A$, will be received by $B$'' \cite{GlabbeekAWN}. The topology-dependent behavior of communication, and
consequently the need for multi-hop communication between nodes, make their properties preconditioned by the existence of some paths among nodes.

To investigate the topology-dependent properties of MANETs by equational reasoning, it is necessary to enrich our process theory {\it RCNT} to specify behaviors constrained by multi-hop constraints.  To this aim, we extend the action prefix operator of {\it RCNT} with actions that are paired with multi-hop constraints, first introduced in \cite{FMSD16} and here extended by negative  multi-hop connectivity pairs. Viewing a network topology as a directed graph, a multi-hop
constraint is represented as a set of multi-hop (dis)connectivity pairs
$\pconn: {\it Loc}\times {\it Loc}$ and $\npconn: {\it Loc}\times {\it Loc}$. For instance, $A\pconn C$
denotes there exists a multi-hop connection from $A$ to $C$, and
consequently $C$ can indirectly receive data from $A$. Let
$\mathbb{M}({\it Loc})$ denote the set of multi-hop constraints that
can be defined over network addresses in ${\it Loc}$, ranged over by
$\mathcal{M}$. Term $(\mathcal{M},\iota).t$, where $\iota\in {\it IAct}\cup\{\tau\}$, denotes that the action $\iota$ is possible if the underlying topology satisfies the multi-hop network constraint $\mathcal{M}$. Formally, a topology like $\gamma$ satisfies the multi-hop network constraint $\mathcal{M}$, denoted by $\gamma\models\mathcal{M} $  iff for each
$\ell\pconn\ell'$ in $\mathcal{M}$, there is a multi-hop connection
from $\ell$ to $\ell'$ in $\gamma$, and for each
$\ell\npconn\ell'$ in $\mathcal{M}$, there is no multi-hop connection
from $\ell$ to $\ell'$ in $\gamma$. To define a well-formed {\it RCNT} term, the rule which restricts the application of the new prefixed-actions to sequential processes, is added to the previous ones. Furthermore, a term cannot have two summands such that one is prefixed by an action of the form $(\mathcal{C},\eta)$ and the other by an action of the form $(\mathcal{M},\iota)$. So terms with an action of the form $(\mathcal{M},\iota)$ only contain action prefix (with multi-hop constraints), choice and recursion operators.

To reason about the correctness of a MANET protocol, its behavior can be abstractly specified by observable internal actions with the required conditions on the underlying topology, i.e., $\iota$-actions with multi-hop constraints. Intuitively, each communication of a protocol implementation triggers an internal action. Such communications are abstracted away by $\tau$-transitions. Therefore, we define a novel preorder relation to examine if a protocol refines its specification. To this aim, a sequence of $\tau$-transitions is allowed to precede an action that is matched to an action of the specification, as long as the accumulated network constraints of the $\tau$-transitions satisfy the multi-hop network constraint of the matched action. Hence our preorder relation is parametrized by a network constraint to reflect such accumulated network constraints.

To provide such a relation, we use the notation $\xRightarrow{\mathcal{C}}$ which is the reflexive and transitive closure of $\tau$-relations while their network constraints are accumulated:\begin{itemize}
\item $t\xRightarrow{\{~\}} t$;
\item if $t\overto{(\mathcal{C},\tau)}t'$ for some arbitrary network constraint $\mathcal{C}$ and
$t'\xRightarrow{\mathcal{C}'}t''$, then $t\xRightarrow{\mathcal{C}'\cup \mathcal{C}} t''$, where $\mathcal{C}'\cup \mathcal{C}$ is well-formed.
\end{itemize}
Furthermore, the network constraint $\mathcal{C}$ satisfies the multi-hop constraint $\mathcal{C}$, denoted by $\mathcal{C}\models \mathcal{M}$ iff
$\exists\, \gamma\in\Gamma(\mathcal{C})\,( \gamma\models\mathcal{M})$. We remark that a network constraint like $\{A\nconn B\}$ may satisfy both $\{A\pconn B\}$ and $\{A\npconn B\}$, but $\{A\conn B\}$ only satisfies $\{A\pconn B\}$.

\begin{definition} {\label{Def::refine}}
A binary relation $\mathcal{R}_\mathcal{C}$ on ${\it RCNT}$ terms is a refinement relation if $t~\mathcal{R}_\mathcal{C}~s$ implies:
\begin{itemize}
\item if $t\overto{(\mathcal{C}',\eta)}t'$, where $\mathcal{C}\cup\mathcal{C}'\in\mathbb{C}^v({\it Loc})$, then\begin{itemize}
\item $\eta=\tau$ and $t'~ \mathcal{R}_{\mathcal{C}\cup\mathcal{C}'}~ s $ with $\mathcal{C}\cup\mathcal{C}'\models \mathcal{M}$, or
\item there is an $s'$ such that $s\overto{(\mathcal{C},\eta)}s'$, and $t'~ \mathcal{R}_{\mathcal{C}\cup\mathcal{C}'}~ s'$, and $\mathcal{C}\cup\mathcal{C}'\models \mathcal{M}$, or
\item $\eta=\iota$ for some $\iota\in{\it IAct}\cup\{\tau\}$ and there is an $s'$ such that $s\overto{(\mathcal{M},\iota)}s'$ with $t'~ \mathcal{R}_{\mathcal{C}\cup\mathcal{C}'}~ s'$;
\end{itemize}
\item if $s\overto{(\mathcal{M},\iota)}s'$, then there are $t''$ and $t'$ such that $t\xRightarrow{\mathcal{C}'}t''\overto{(\mathcal{C}'',\iota)}t'$ with $t''~ \mathcal{R}_{\mathcal{C}\cup\mathcal{C}'}~ s$ and $t'~ \mathcal{R}_{\mathcal{C}\cup\mathcal{C}'\cup\mathcal{C}''}~ s'$;
\item if $s\overto{(\mathcal{C},\eta)}s'$, then there is a $t'$ such that $t\overto{(\mathcal{C}',\eta)}t'$ with  $t'~ \mathcal{R}_{\mathcal{C}\cup\mathcal{C}'}~ s'$.
\end{itemize}
The protocol $t$ refines the specification $s$, denoted by $t\sqsubseteq s$, if $t~\mathcal{R}_{\{~\}}~s$ for some refinement relation $\mathcal{R}_{\{~\}}$.
\end{definition}

\begin{theorem}\label{Th::preorder}
Refinement is a preorder relation and has the precongruence property.
\end{theorem}
See Section \ref{sec::preorder} for its proof. To analyze the correctness of our simple routing protocol, we investigate if it has the packet delivery property. To this end, we verify whether $\tau_\Msg(\partial_\Msg(\oc P\cc_A\parallel \oc M\cc_C\parallel \oc Q\cc_B))$ refines $\mathfrak{S}$, where $\mathfrak{S}\deff(\{A\pconn B,B\pconn A\},\deliver).\mathfrak{S}+(\{A\npconn B\},\tau).0+(\{A\pconn B,B\npconn A\},\tau).0$.
To this aim, we match all the resulting terms of  $\tau$-transitions to $\mathfrak{S}$ as long as their accumulated network constraints satisfy $\{A\pconn B,B\pconn A\}$. If a $\tau$-transition violates $\{A\pconn B,B\pconn A\}$ but satisfies $\{A\npconn B\}$, then it will be matched to the transition $(\{A\npconn B\},\tau)$. Otherwise, it will be matched to the transition $(\{A\pconn B,B\npconn A\},\tau)$. Therefore, we exploit the provided equations together with the precongruence property of our refinement for the choice operator and the rules of Proposition \ref{Pro::rules}.

\begin{proposition}\label{Pro::rules}
Suppose $\iota\in{\it IAct}$. The following rules holds\[\begin{array}{l}(\mathcal{C},\tau).t\sqsubseteq (\mathcal{M},\iota).s \Leftrightarrow \mathcal{C}\rhd t \sqsubseteq (\mathcal{M},\iota).s \,\wedge\,\mathcal{C}\models\mathcal{M}\\
(\mathcal{C},\iota).t\sqsubseteq (\mathcal{M},\iota).s \Leftrightarrow \mathcal{C}\rhd t \sqsubseteq s
\end{array}
\]
\end{proposition} These rules correspond to the transfer conditions of Definition \ref{Def::refine}, and their proofs are discussed in Section \ref{sec::preorder}.

Thus, we use Equation \ref{eq::base} to show that:
{{\allowdisplaybreaks
\begin{flalign}
&\tau_\Msg(\partial_\Msg(\oc P\cc_A\parallel \oc M\cc_C\parallel \oc Q\cc_B))~\sqsubseteq~\mathfrak{S} \Leftrightarrow\nonumber\\ &
\{A\conn B\}\rhd\tau_\Msg(\partial_\Msg(\oc P\cc_A\parallel \oc M\cc_C\parallel \oc \deliver.Q\cc_B))~\sqsubseteq~\mathfrak{S}\,\wedge\nonumber\\
&\{A\nconn B,A\conn C\}\rhd\tau_\Msg(\partial_\Msg(\oc P_1\cc_A\parallel \oc \snd(\req).M_1\cc_C\parallel \oc Q\cc_B))~\sqsubseteq~\mathfrak{S}\,\wedge\label{eq:main}\\
&\{A\nconn B,A\nconn C\}\rhd\tau_\Msg(\partial_\Msg(\oc P_1\cc_A\parallel \oc M\cc_C\parallel \oc Q\cc_B))~\sqsubseteq~0\nonumber
\end{flalign}}%
} 
To prove the refinement relation \ref{eq:main}, we use the Equation \ref{eq::m} to show that
{{\allowdisplaybreaks
\begin{flalign}
& \{A\nconn B,A\conn C\}\rhd \tau_\Msg(\partial_\Msg(P_1\cc_A\parallel \oc\snd(\req).M_1\cc_C\parallel \oc Q\cc_B))~\sqsubseteq~\mathfrak{S}\Leftrightarrow\nonumber\\ 
&\{A\nconn B,A\conn C,C\conn B\}\rhd\tau_\Msg(\partial_\Msg(\oc P_1\cc_A\parallel \oc M_1\cc_C\parallel \oc \snd(\rep_B).Q\cc_B))~\sqsubseteq~{\mathfrak{S}}\,\wedge\ \nonumber\\
&\{A\nconn B,A\conn C,C\nconn B\}\rhd\tau_\Msg(\partial_\Msg(\oc P_1\cc_A\parallel \oc M_1\cc_C\parallel \oc Q\cc_B))~\sqsubseteq~0 \label{Eq::stop}
\end{flalign}}%
}This proof process stops when we reach to the predicate $\mathcal{C}\rhd t~\sqsubseteq~(\mathcal{M},\iota).s$ to prove for which either we have previously examined $\mathcal{C}'\rhd t~\sqsubseteq~(\mathcal{M},\iota).s$ where $\mathcal{C}\preccurlyeq\mathcal{C}'$, or it holds trivially. For instance, the refinement relation (\ref{Eq::stop}) trivially holds as it can be proved with the help of our axiomatization, especially the rules ${\it Fold}$ and ${\it TRes}_{1,2}$, that $\{A\nconn B,A\conn C,C\nconn B\}\rhd\tau_\Msg(\partial_\Msg(\oc P_1\cc_A\parallel \oc M_1\cc_C\parallel \oc Q\cc_B))$ is the answer to the equation $\mathfrak{Q}\deff(\{A\nconn B,A\conn C,C\nconn B\},\tau).\mathfrak{Q}$, and trivially
\[\rec\mathfrak{Q}\cdot(\{A\nconn B,A\conn C,C\nconn B\},\tau).\mathfrak{Q}~\sqsubseteq~0.\]
So, it can be easily proved that $\tau_\Msg(\partial_\Msg(\oc P\cc_A\parallel \oc M\cc_C\parallel \oc Q\cc_B))~\sqsubseteq~\mathfrak{S}$
.

\section{Related Work}\label{sec::concreteRelatedWork}

Related calculi to ours are CBS\#~\cite{CBSNanz}, CWS~\cite{CWS},
CMAN~\cite{CMAN,Godskesen08}, CMN~\cite{CMN} and its timed version \cite{MerroBS11}, bKlaim~\cite{bKlaim}, $\omega$-calculus~\cite{w-cal}, SCWN ~\cite{GodskesenN09}, CSDT~\cite{CDST},
AWN~\cite{GlabbeekAWN} and its timed extension \cite{TimedAWN}, and the broadcast psi-calculi
\cite{psi}. These approaches have already been compared in \cite{FatemehFI10} with regard to modeling issues, such as topology and mobility, as well as behavioral congruence relations, in particular observables and distinguishing power. As all these approaches, except ~\cite{CWS}, focus on protocols above the data like layer, we investigate their capabilities to faithfully support the properties of wireless communication at this layer, i.e., being non-blocking and asynchronous. Furthermore, we compare our behavioral equivalence relation to those with a reliable setting.

All these approaches, except \cite{wrebeca}, provide an algebraic framework. Among them only \cite{bKlaim} is asynchronous, centered around the tuple space paradigm; broadcast messages are output into the tuple spaces of neighboring nodes to the sending node.

The non-blocking property is a consequence of either nodes being \textit{input-enabled} or the communication primitives being lossy. In the former case, the asynchronous property is achieved through abstract data specifications~\cite{ADT} in line with the approach from~\cite{GP91,GP95}, in which the sum operator plays a pivotal role. Each process is then parametrized by a variable of the queue type with a summand which receives all possible messages (if the queue is empty). Among these approaches, CMN, CMAN, $\omega$-calculus, SCWN, and the broadcast psi-calculi are lossy.

To make a process input-enabled while communication is synchronous, three approaches are followed. In the first approach, followed by AWN, the semantics is equipped with a rule similar to our ${\it Rcv}_2$ with a negative premise which expresses that if a node is not ready to receive, the message is simply ignored \cite{GlabbeekAWN}. Due to our implicit modeling of topology, the negative premise of our rule is more complicated to characterize the unreadiness of nodes regarding the underlying topology. In the second approach, followed by CDST, counterparts for the rules $\it Bro$ and $\it Recv$ are defined with negative premises to cover cases when a process cannot participate in the communication message \cite{CDST}. The third approach, provided by CSB\#, eliminates negative premises, to remain within the de Simone format of structural operational semantics~\cite{GlabbeekAWN}, in favor of actions which discard messages \cite{CBSNanz}. Therefore, the semantics is augmented by rules that trigger the ignore actions for any sending node, receiving nodes for disconnected locations, and deadlock. Furthermore, the rules $\it Bro$ and $\it Recv$ are modified to cover cases when a process ignores a message. 

Among the reliable settings, only CDST provides a behavioral equivalence relation, based on the notion of observational congruence: the receive and send actions are observable while transitions changing the underlying topology are treated as unobservable. However, due to implicit modeling of topology and mobility, our behavioral equivalence relation has been parametrized with network constraints while it considers the branching structure of MANETs.

\section{Conclusion}
We introduced the reliable framework {\it RRBPT}, suitable to specify and verify MANETs, with the aim to catch errors in design decisions. We discussed the required changes at the semantic model by extending the network constraints with negative connectivity links. Furthermore, we revised the equivalence relation of the lossy setting to preserve required behavior in the reliable framework. Then we demonstrated which axioms should be added to /removed from the reliable setting. We provided an analysis approach at the syntactic level, exploiting a precongruence relation and our axiomatization. We applied our analysis approach to a simple routing protocol to prove that it correctly finds routes among connected nodes.

\nocite{*}
\bibliographystyle{fundam}


\appendix

\section{\label{sec::app}Branching Reliable Computed Network Bisimilarity is an Equivalence} To prove that branching reliable computed network bisimilarity
is an equivalence, we exploit semi-branching reliable computed network
bisimilarity, following \cite{Basten}. 

\begin{definition} {\label{Def::sbbism}} A binary relation
$\mathcal{R}$ on computed network terms is a semi-branching reliable computed
network simulation, if $t_1\mathcal{R} t_2$
implies whenever $t_1\overto{(\mathcal{C},\eta)} t_1'$:
\begin{itemize}
\item either $\eta=\tau$ and there is a $t_2'$ such that $t_2\Rightarrow t_2'$ with $t_1\mathcal{R}t_2'$ and $t_1'\mathcal{R}t_2'$; or
\item there are $s_1'',\ldots,s_k''$ and $s_1',\ldots,s_k'$ for some $k>0$
such that $\forall i\le k\,(t_2\Rightarrow s_i''\overto{\langle(\mathcal{C}_i,\eta)\rangle} s_i'$, with $t_1\mathcal{R} s_i''$ and $t_1'\mathcal{R}
s_i')$, and $\langle\mathcal{C}_1\rangle,\ldots,\langle\mathcal{C}_k\rangle$ constitute a partitioning of $\langle\mathcal{C}\rangle$.
\end{itemize} $\mathcal{R}$ is a semi-branching reliable computed network bisimulation if $\mathcal{R}$ and ${\mathcal{R}}^{-1}$
are semi-branching reliable computed network simulations. Computed networks
$t_1$ and $t_2$ are semi-branching reliable computed
network bisimilar if $t_1\mathcal{R} t_2$, for
some semi-branching reliable computed network bisimulation relation
$\mathcal{R}$.
\end{definition}

\begin{lemma}\label{lem::property}
Let $t_1$ and $t_2$ be computed network terms,
and $\mathcal{R}$ a semi-branching reliable computed network bisimulation
such that $t_1 \mathcal{R} t_2$.
\begin{itemize}
\item If $t_1\Rightarrow t_1'$ then $\exists
t_2'\cdot t_2\Rightarrow
t_2'\wedge t_1' \mathcal{R} t_2'$
\item If $t_2\Rightarrow t_2'$ then $\exists
t_1'\cdot t_1\Rightarrow
t_1'\wedge t_1' \mathcal{R} t_2'$
\end{itemize}
\end{lemma}
\begin{proof}We only give the proof of the first property. The second
property can be proved in a similar fashion. The proof is by
induction on the number of $\Rightarrow$ steps from $t_1$
to $t_1'$:\begin{itemize}
\item Base: Assume that the number of steps equals zero. Then
$t_1$ and $t_1'$ must be equal. Since
$t_1 \mathcal{R} t_2$ and
$t_2\Rightarrow t_2$, the property is satisfied.
\item Induction step: Assume $t_1\Rightarrow
t_1'$ in $n$ steps, for some $n\ge 1$. Then there is  $t_1''$ such that $t_1\Rightarrow t_1''$
in $n-1$ $\Rightarrow$ steps, and $t_1''\overto{(\mathcal{C},\tau)}t_1'$. By the induction hypothesis, there is a $t_2''$ such that $t_2\Rightarrow t_2''$ and $t_1'' \mathcal{R} t_2''$.
Since $t_1''\overto{(\mathcal{C},\tau)}t_1'$ and $\mathcal{R}$ is a semi-branching reliable computed network bisimulation,
there are two cases to consider: \begin{itemize}
\item there is a $t_2'$ such that
$t_2''\Rightarrow t_2'$,
$t_1'' \mathcal{R} t_2'$, and $t_1' \mathcal{R}
t_2'$. So $t_2\Rightarrow
t_2'$ such that $t_1' \mathcal{R}
t_2'$.
\item or there are $s_1''',\ldots,s_k'''$ and $s_1',\ldots,s_k'$ for some $k>0$
such that $\forall i\le k \,(t_2''\Rightarrow s_i'''\overto{(\mathcal{C}_i,\tau)} s_i'$,
with $t_1''\mathcal{R} s_i'''$ and $t_1'\mathcal{R}
s_i')$, and $\mathcal{C}_1,\ldots,\mathcal{C}_k$ constitute a partitioning of $\mathcal{C}$. By definition, $s_i'''\overto{(\mathcal{C}_i,\tau)}
s_i'$ yields $s_i'''\Rightarrow
s_i'$. Consequently for any arbitrary $i\le k$, $t_2\Rightarrow
s_i'$ such that $t_1' \mathcal{R}
s_i'$.\end{itemize}
\end{itemize}
\end{proof}

\begin{proposition}\label{Pro::compose}
The relation composition of two semi-branching reliable computed network
bisimulations is again a semi-branching reliable computed network
bisimulation.
\end{proposition}
\begin{proof}Let $\mathcal{R}_1$ and $\mathcal{R}_2$ be
semi-branching reliable computed network bisimulations with $t_1
\mathcal{R}_1 t_2$ and $t_2 \mathcal{R}_2
t_3$. Let $t_1\overto{(\mathcal{C},\eta)} t_1'$.
It must be shown that\begin{itemize}
\item either $\eta=\tau$ and there is a $t_3'$ such that $t_3\Rightarrow t_3'$ with $t_1\mathcal{R}_1\circ \mathcal{R}_2 t_3'$ and $t_1'\mathcal{R}_1\circ \mathcal{R}_2 t_3'$; or
\item $\exists s_1',\ldots,s_k',s_1'',\ldots,s_k''\,\forall i\le k\,(t_3\Rightarrow s_i''\overto{\langle (\mathcal{C}_i,\eta)\rangle}s_i'\,\wedge\, t_1 \mathcal{R}_1\circ \mathcal{R}_2 s_i'' \,\wedge\, t_1' \mathcal{R}_1\circ \mathcal{R}_2 s_i')$, where $\langle\mathcal{C}_1\rangle,\ldots,\langle\mathcal{C}_k\rangle$ constitute a partitioning of $\langle\mathcal{C}\rangle$.
\end{itemize}

Since $t_1 \mathcal{R}_1 t_2$, two cases can be considered:

\begin{itemize}
\item $\eta=\tau$ and there is a $t_2'$ such that $t_2\Rightarrow t_2'$ with $t_1 \mathcal{R}_1 t_2'$ and $t_1' \mathcal{R}_1 t_2'$. Lemma~\ref{lem::property} yields that there is a $t_3'$ that $t_3\Rightarrow t_3'$ with $t_2' \mathcal{R}_2 t_3'$. It immediately follows that $t_1\mathcal{R}_1\circ \mathcal{R}_2 t_3'$ and $t_1'\mathcal{R}_1\circ \mathcal{R}_2 t_3'$.
\item there exist
$s^{\ast\ast}_1,\ldots s^{\ast\ast}_j$, $s^{\ast}_1\ldots s^{\ast}_j$ for some $j>0$ such that
$\forall i\le j\,(t_2\Rightarrow s^{\ast\ast}_i\overto{\langle(\mathcal{C}_i,\eta)\rangle}
s^{\ast}_i$, $t_1 \mathcal{R}_1 s^{\ast\ast}_i$,
$t_1' \mathcal{R}_1 s^{\ast}_i)$, and $\langle\mathcal{C}_1\rangle,\ldots,\langle\mathcal{C}_j\rangle$ is a partitioning of $\langle\mathcal{C}\rangle$. Since $t_2
\mathcal{R}_2 t_3$ and $t_2\Rightarrow
s^{\ast\ast}_i$, Lemma~\ref{lem::property} yields that there are
$s_1''',\ldots,s_j'''$ such that $\forall i\le j\,(t_3\Rightarrow
s_i'''\,\wedge \,s^{\ast\ast}_i \mathcal{R}_2
s_i''')$. Two cases can be distinguished:\begin{itemize}
\item either $\eta=\tau$ and for some $i\le j $, $s^{\ast\ast}_i\overto{(\mathcal{C}_i,\tau)}
s^{\ast}_i$ implies there is an $s_i''$ such that $s_i'''\Rightarrow s_i''$ with $s^{\ast\ast}_i \mathcal{R}_2 s_i''$ and $s^{\ast}_i\mathcal{R}_2 s_i''$. It follows immediately that there is an $s_i''$ such that $t_3\Rightarrow s_i''$ with $t_1 \mathcal{R}_1\circ \mathcal{R}_2 s_i''$ and $t_1'\mathcal{R}_1\circ \mathcal{R}_2 s_i''$; or
\item for all $i\le j $, $s^{\ast\ast}_i\overto{\langle(\mathcal{C}_i,\eta)\rangle} s^{\ast}_i$ implies there are  $s_{i_1}'',\ldots,s_{i_{k_i}}''$ and $s_{i_1}',\ldots,s_{i_{k_i}}'$ for some $k_i>0$ such that $\forall o\le k_i\,(s_i''\Rightarrow
s_{i_o}''\overto{\langle(\mathcal{C}_{i_o},\eta)\rangle} s_{i_o}'$,
$s^{\ast\ast}_i \mathcal{R}_2  s_{i_o}''$, $s^{\ast}_i \mathcal{R}_2
s_{i_o}')$, and $\langle\mathcal{C}_{i_1}\rangle,\ldots,\langle\mathcal{C}_{i_{k_i}}\rangle$ is a partitioning of $\langle\mathcal{C}_i\rangle$. Since $t_3\Rightarrow s_i''$, we have $\forall i\le j,\, \forall o\le k_i\,(t_3\Rightarrow
s_{i_o}''\overto{\langle(\mathcal{C}_{i_o},\eta)\rangle} s_{i_o}'$ with $t_1 \mathcal{R}_1\circ\mathcal{R}_2 s_{i_o}''$, $t_1' \mathcal{R}_1\circ\mathcal{R}_2 s_{i_o}')$, and $\{\langle\mathcal{C}_{i_o}\rangle\,\mid\,i\le j,o<k_i\}$ is a partitioning of $\langle\mathcal{C}\rangle$.
\end{itemize}
\end{itemize}
\end{proof}

\begin{corollary}\label{Cor::semieq}
Semi-branching reliable computed network bisimilarity is an equivalence relation.
\end{corollary}

It is not hard to see that the union of semi-branching reliable computed network bisimulations is again a semi-branching reliable computed network bisimulation.

\begin{proposition}\label{Prop::semi-branching}
The largest semi-branching reliable computed network bisimulation is a branching reliable computed network bisimulation.
\end{proposition}

\begin{proof}Suppose $\mathcal{R}$ is the largest semi-branching reliable
computed network bisimulation for some given CLTS. Let $ t_{1}\mathcal{R} t_2$,
$ t_2\Rightarrow t_2'$, $ t_{1}
\mathcal{R} t_2'$ and $ t_{1}' \mathcal{R}
t_2'$. We show that
$\mathcal{R}'=\mathcal{R}\cup\{( t_{1}', t_2)\}$
is a semi-branching reliable computed network bisimulation.
\begin{enumerate}
\item If $ t_1'\overto{(\mathcal{C},\eta)} t_1''$, then it follows from $( t_1', t_2')\in
\mathcal{R}$ that\begin{itemize}
\item either $\eta=\tau$ and there is a $t_2''$ such that $t_2'\Rightarrow t_2''$ with $t_1'\mathcal{R}t_2''$ and $t_1''\mathcal{R}t_2''$. Finally $t_2\Rightarrow t_2'$ results $t_1'\mathcal{R}'t_2''$ and $t_1''\mathcal{R}'t_2''$; or
\item there are $ s_1''',\ldots,s_k'''$ and $ s_1'',\ldots,s_k''$ for some $k>0$ such that
$\forall i\le k\,( t_2'\Rightarrow s_i'''\overto{\langle(\mathcal{C}_i,\eta)\rangle} s_i''$ with $( t_1', s_i'''),( t_1'', s_i'')\in \mathcal{R})$ and $\langle\mathcal{C}_1\rangle,\ldots,\langle\mathcal{C}_k\rangle$ is a partitioning of $\langle\mathcal{C}\rangle$. And $ t_2 \Rightarrow  t_2'$ yields $\forall i\le k\,( t_{2}\Rightarrow  s_i'''\overto{\langle(\mathcal{C}_i,\eta)\rangle} s_i''$, with $( t_1', s_i'''),( t_1'', s_i'')\in \mathcal{R}')$.
\end{itemize}
\item If $ t_2\overto{(\mathcal{C},\eta)} t_{2}''$,
then it follows from $( t_1, t_2)\in
\mathcal{R}$ that\begin{itemize}
\item either $\eta=\tau$, and there is a $t_1''$ such that $t_1\Rightarrow t_1''$ with $t_1''\mathcal{R}t_2$ and $t_1''\mathcal{R}t_2''$. Furthermore, $(t_1,t_2')\in\mathcal{R}$, $t_1\Rightarrow t_1''$, and Lemma \ref{lem::property} imply there is a $t_2'''$ such that $t_2'\Rightarrow t_2'''$ with $(t_1'',t_2''')\in\mathcal{R}$. Similarly $(t_1',t_2')\in\mathcal{R}$, $t_2'\Rightarrow t_2'''$, and Lemma \ref{lem::property} imply there is a $t_1'''$ such that $t_1'\Rightarrow t_1'''$ with $(t_1''',t_2''')\in\mathcal{R}$. From $(t_1''',t_2''')\in\mathcal{R}$, $(t_2''',t_1'')\in\mathcal{R}^{-1}$, and $(t_1'',t_2)\in\mathcal{R}$, we conclude $(t_1''',t_2)\in \mathcal{R}\circ\mathcal{R}^{-1}\circ\mathcal{R}$. And from $(t_1''',t_2''')\in\mathcal{R}$, $(t_2''',t_1'')\in\mathcal{R}^{-1}$, and $(t_1'',t_2'')\in\mathcal{R}$, we conclude $(t_1''',t_2'')\in \mathcal{R}\circ\mathcal{R}^{-1}\circ\mathcal{R}$.
\item or there are $ s_{1_1}''',\ldots,s_{1_k}'''$
and $ s_{1_1}'',\ldots,s_{1_k}''$ for some $k>0$ such that $\forall i\le k\,( t_1\Rightarrow
s_{1_i}'''\overto{\langle(\mathcal{C}_i,\eta)\rangle} s_{1_i}''$
with $( s_{1_i}''', t_2),( s_{1_i}'', t_2'')\in
\mathcal{R})$ and $\langle\mathcal{C}_1\rangle,\ldots,\langle\mathcal{C}_k\rangle$ is a partitioning of $\langle\mathcal{C}\rangle$. Since $( t_1, t_2')\in
\mathcal{R}$ and $ t_1\Rightarrow
s_{1_i}'''$, by Lemma~\ref{lem::property},
there are $s_{2_1}''',\ldots,s_{2_k}'''$ such that
$\forall i\le k\,( t_2'\Rightarrow
s_{2_i}'''$ and $( s_{1_i}''',s_{2_i}''')\in
\mathcal{R})$. Since
$ s_{1_i}'''\overto{\langle(\mathcal{C}_i,\eta)\rangle} s_{1_i}''$,
there are $s_{2_{i_1}}^{\ast\ast},\ldots,s_{2_{i_{k_i}}}^{\ast\ast}$ and
$s_{2_{i_1}}^{\ast},\ldots,s_{2_{i_{k_i}}}^{\ast}$ for some $k_i>0$ such that
$\forall o\le k_i\,(s_{2_i}'''\Rightarrow
s_{2_{i_o}}^{\ast\ast}\overto{\langle(\mathcal{C}_{i_o},\eta)\rangle}
s_{2_{i_o}}^{\ast}$ with $( s_{1_i}''',s_{2_{i_o}}^{\ast\ast}),(s_{1_i}'',s_{2_{i_o}}^{\ast})\in
\mathcal{R})$ and $\langle\mathcal{C}_{i_1}\rangle,\ldots,\langle\mathcal{C}_{i_{k_i}}\rangle$ is a partitioning of $\langle\mathcal{C}_i\rangle$. Since $ t_{2}'\Rightarrow
s_{2_i}'''$ and $s_{2_i}'''\Rightarrow
s_{2_{i_o}}^{\ast\ast}$, we have
$\forall i\le k,o\le k_i\,( t_{2}'\Rightarrow s_{2_{i_o}}^{\ast\ast})$.
By assumption, $( t_1', t_2')\in
\mathcal{R}$, so by Lemma \ref{lem::property} there are $s_{1_1}^{\ast\ast},\ldots,s_{1_K}^{\ast\ast}$, where $K=\sum_{i=1}^k k_i$, such that
$\forall z\le K\,( t_1'\Rightarrow  s_{1_z}^{\ast\ast}$ and $(s_{1_z}^{\ast\ast}, s_{2_{i_o}}^{\ast\ast})\in
\mathcal{R}$, where $z=(\sum_{j=1}^{i-1}k_j)+o)$. Since $s_{2_{i_o}}^{\ast\ast}\overto{\langle(\mathcal{C}_{i_o},\eta)\rangle} s_{2_{i_o}}^{\ast}$, there are
$s_{1_{z_1}}^{\ast\ast\ast},\ldots,s_{1_{z_{k'_z}}}^{\ast\ast\ast}$ and $s_{1_{z_1}}^{'},\ldots,s_{1_{z_{k_z'}}}^{'}$ for some $k'_z>0$
such that $\forall j\le k'_z\,( s_{1_z}^{\ast\ast}\Rightarrow
s_{1_{z_j}}^{\ast\ast\ast}\overto{\langle(\mathcal{C}_{i_{o_j}},\eta)\rangle} s_{1_{z_j}}^{'}$
with $(s_{1_{z_j}}^{\ast\ast\ast},s_{2_{i_o}}^{\ast\ast}),(s_{1_{z_j}}^{'},s_{2_{i_o}}^{\ast})\in
\mathcal{R})$ and $\langle\mathcal{C}_{i_{o_1}}\rangle,\ldots,\langle\mathcal{C}_{i_{o_{k_z'}}}\rangle$ is a partitioning of $\langle\mathcal{C}_{i_o}\rangle$. And $ t_1'\Rightarrow s_{1_z}^{\ast\ast}$ yields $ \forall
i\le k,o\le k_i,j\le k_z'(t_1'\Rightarrow
s_{1_{z_j}}^{\ast\ast\ast}\overto{\langle(\mathcal{C}_{i_{o_j}},\eta)\rangle} s_{1_{z_j}}^{'}$ with \[
\begin{array}{c}
(s_{1_{z_j}}^{\ast\ast\ast},s_{2_{i_o}}^{\ast\ast})\in
\mathcal{R} \wedge (s_{2_{i_o}}^{\ast\ast},s_{1_i}''')\in
\mathcal{R}^{-1}\wedge (s_{1_i}''', t_2)\in
\mathcal{R} \\ \hspace*{2cm}\Rightarrow
(s_{1_{z_j}}^{\ast\ast\ast}, t_2)\in
\mathcal{R}\circ \mathcal{R}^{-1}\circ \mathcal{R}\\
(s_{1_{z_j}}^{'},s_{2_{i_o}}^{\ast})\in
\mathcal{R} \wedge (s_{2_{i_o}}^{\ast}, s_{1_i}'')\in
\mathcal{R}^{-1}\wedge (  s_{1_i}'', t_1'')\in
\mathcal{R} \\ \hspace*{2cm}\Rightarrow
(s_{1_{z_j}}^{'}, t_2'')\in
\mathcal{R}\circ \mathcal{R}^{-1}\circ \mathcal{R},
\end{array}
\]where $z=(\sum_{l=1}^{i-1}k_l)+o)$, and $\{\langle\mathcal{C}_{i_{o_j}}\rangle\,\mid\, i\le k, o\le k_i, j\le k_z'\}$ is a partitioning of $\langle\mathcal{C}\rangle$.
\end{itemize}
By Proposition~\ref{Pro::compose}, $\mathcal{R}\circ \mathcal{R}^{-1}\circ \mathcal{R}$ is a semi-branching reliable
computed network bisimulation. Since $\mathcal{R}$ is the largest semi-branching reliable computed network bisimulation, and clearly $\mathcal{R} \subseteq \mathcal{R}\circ \mathcal{R}^{-1}\circ \mathcal{R}$, we have $\mathcal{R}=\mathcal{R}\circ \mathcal{R}^{-1}\circ
\mathcal{R}$.
\end{enumerate}
So $\mathcal{R}'$ is a semi-branching reliable computed network bisimulation. Since $\mathcal{R}$ is the largest semi-branching reliable
computed network bisimulation, $\mathcal{R}'=\mathcal{R}$.

We will now prove that $\mathcal{R}$ is a branching reliable computed network
bisimulation. Let $ t_1 \mathcal{R}  t_2$, and
$ t_1\overto{(\mathcal{C},\eta)}  t_1'$. We only consider the
case when $\eta=\tau$, because for other cases, the transfer condition of
Definition~\ref{Def::brbism} and Definition~\ref{Def::sbbism} are the
same.
Two cases can be
distinguished:\begin{enumerate}
\item There is a $t_2'$ such that $t_2\Rightarrow t_2'$ with $t_1\mathcal{R}t_2'$ and $t_1'\mathcal{R}t_2'$:  we proved above that $ t_1' \mathcal{R}
t_2$. This agrees with the first case of
Definition~\ref{Def::brbism}.
\item There are $ s_1'',\ldots,s_k''$ and $ s_1',\ldots,s_k'$ for some $k>0$ such that
$\forall i\le k\, (t_2\Rightarrow
s_i''\overto{\langle(\mathcal{C}_{i},\tau)\rangle} s_i'$ with
$ t_1 \mathcal{R}  s_i''$ and $ t_1'
\mathcal{R}  s_i')$ and $\langle\mathcal{C}_1\rangle,\ldots,\langle\mathcal{C}_k\rangle$ constitute a partitioning of $\langle\mathcal{C}\rangle$. This agrees with
the second case of Definition~\ref{Def::bbism}.\end{enumerate}
Consequently $\mathcal{R}$ is a branching reliable computed network
bisimulation.
\end{proof}

Since any branching reliable computed network bisimulation is a
semi-branching reliable computed network bisimulation, this yields the
following corollary.

\begin{corollary}\label{Co::relate}
Two computed network terms are related by a branching reliable computed
network bisimulation if and only if they are related by a
semi-branching reliable computed network bisimulation.
\end{corollary}

\begin{corollary}\label{Co::branchingeq}
Branching reliable computed network bisimilarity is an equivalence relation.
\end{corollary}

\begin{corollary}\label{Co::rootedbranchingeq}
Rooted branching reliable computed network bisimilarity is an equivalence relation.
\end{corollary}

\begin{proof}
It is easy to show that rooted branching reliable computed network
bisimilarity is reflexive and symmetric. To conclude the proof, we
show that rooted branching reliable computed network bisimilarity is
transitive. Let $ t_1\simeq_{rbr} t_2$ and
$ t_2\simeq_{rbr} t_3$. Since
$ t_1\simeq_{rbr} t_2$, if
$ t_1\overto{(\mathcal{C},\eta)} t_1'$, then there is
$ t_2'$ such that
$ t_2\overto{\langle(\mathcal{C},\eta)\rangle} t_2'$ and
$ t_1'\simeq_{br} t_2'$. Since
$ t_2\simeq_{rbr} t_3$, there is a $ t_3'$
such that $ t_3\overto{\langle(\mathcal{C},\eta)\rangle} t_3'$ and
$ t_2'\simeq_{br} t_3'$. Equivalence of branching reliable
computed network bisimilarity yields
$ t_3\overto{\langle(\mathcal{C},\eta)\rangle} t_3'$ with
$ t_1'\simeq_{br} t_3'$. The same argumentation holds
when $ t_3\overto{(\mathcal{C},\eta)} t_3'$. Consequently
the transfer conditions of Definition~\ref{Def::rbrbism} holds and
$ t_1\simeq_{rbr} t_3$.
\end{proof}

\section{\label{sec::cong}Rooted Branching Reliable Computed Network Bisimilarity is a Congruence}

\begin{theorem}\label{Theo::congruence}
Rooted branching reliable computed network bisimilarity is a congruence for
terms with
respect to {\it RCNT} operators.
\end{theorem}

\begin{proof}We need to prove the following cases:\begin{enumerate}
\item $\oc t_1\cc_\ell\simeq_{rbr} \oc t_2\cc_\ell$ implies $\oc \alpha.t_1\cc_\ell\simeq_{rbr} \oc
\alpha.t_2\cc_\ell$; \label{case::1}
\item $\oc t_1\cc_\ell\simeq_{rbr} \oc t_2\cc_\ell$ and $\oc t_1'\cc_\ell\simeq_{rbr} \oc t_2'\cc_\ell$ implies $\oc t_1+t_1'\cc_\ell\simeq_{rbr} \oc
t_2+t_2'\cc_\ell$; \label{case::2}
\item $\oc t_1\cc_\ell\simeq_{rbr} \oc t_2\cc_\ell$ and $\oc t_1'\cc_\ell\simeq_{rbr} \oc t_2'\cc_\ell$ implies $\oc
{\it sense}(\ell',t_1,t_1')\cc_\ell\simeq_{rbr} \oc{\it sense}(\ell',t_2,t_2')\cc_\ell$;
\label{case::3}
\item $\oc t_1\cc_\ell\simeq_{rbr}\oc t_2\cc_\ell$ implies
$\ell :t:  t_1\simeq_{rbr}\ell :t: t_2$ for any arbitrary term $t$;\label{case::13}
\item $t_1\simeq_{rbr}t_2$ implies
$(\mathcal{C},\eta).t_1\simeq_{rbr}(\mathcal{C},\eta).t_2$; \label{case::4}
\item $ t_1\simeq_{rbr} t_2$ and $ t_1'\simeq_{rbr} t_2'$ implies
$ t_1+ t_1'\simeq_{rbr} t_2+ t_2'$;
\label{case::5}
\item $t_1\simeq_{rbr}t_2$ implies
$(\nu\ell)t_1\simeq_{rbr}(\nu\ell)t_2$;
\label{case::6}
\item $t_1\simeq_{rbr}t_2$ and $t_1'\simeq_{rbr}t_2'$ implies
$t_1\parallel t_1'\simeq_{rbr}t_2\parallel t_2'$;
\label{case::7}
\item $ t_1\simeq_{rbr} t_2$ and $ t_1'\simeq_{rbr} t_2'$ implies
$ t_1\lm t_1'\simeq_{rbr} t_2\lm t_2'$;
\label{case::8}
\item $ t_1\simeq_{rbr} t_2$ and $ t_1'\simeq_{rbr} t_2'$ implies
$ t_1\mid t_1'\simeq_{rbr} t_2\mid t_2'$;
\label{case::9}
\item $t_1\simeq_{rbr}t_2$ implies $\partial_M(t_1)\simeq_{rbr}\partial_M(t_2)$; \label{case::10}
\item $t_1\simeq_{rbr}t_2$ implies
$\tau_M(t_1)\simeq_{rbr}\tau_M(t_2)$;\label{case::11}
\item $t_1\simeq_{rbr}t_2$ implies
$\mathcal{C}\rhd t_1\simeq_{rbr}\mathcal{C}\rhd t_2$.
\label{case::12}
\end{enumerate}

Clearly, if $ t_1\simeq_{rbr} t_2$ then
$ t_1\simeq_{br} t_2$ is witnessed by the following
branching reliable computed network bisimulation relation:\[\begin{array}{l}
\mathcal{R}'=\{\mathcal{R}\,\mid\, t_1\overto{(\mathcal{C},\eta)} t_1'\Rightarrow
\exists
t_2'\cdot t_2\overto{\langle (\mathcal{C},\eta)\rangle } t_2'\wedge
t_1'\simeq_{br} t_2'\mbox{\small ~is witnessed by
$\mathcal{R}$}\} \\
\hspace*{0.5cm}\cup\,\{\mathcal{R}\,\mid\, t_2\overto{(\mathcal{C},\eta)} t_2'\Rightarrow
\exists
t_1'\cdot t_1\overto{\langle (\mathcal{C},\eta)\rangle} t_1'\wedge
t_1'\simeq_{br} t_2'\mbox{\small ~is witnessed by
$\mathcal{R}$}\} \\
\hspace*{0.5cm}\cup\,\{( t_1, t_2)\}.\end{array}\]

We prove the cases
\ref{case::1},~\ref{case::2},~\ref{case::13},~\ref{case::6},~\ref{case::9},~
\ref{case::10}, and~\ref{case::12} since the proof of the cases \ref{case::3} and \ref{case::5} are similar to the case \ref{case::2}, the case~\ref{case::4} is
similar to the case \ref{case::1},
the cases~\ref{case::7} and~\ref{case::8} are
similar to the case \ref{case::9},
and the case~\ref{case::11} is similar to the
case~\ref{case::10}.

\noindent \textbf{Case~\ref{case::1}}. The first transitions of $\oc
\alpha.t_1\cc_\ell$ and $\oc \alpha.t_2\cc_\ell$ are the same with application of the rule ${\it Snd}$ (if $\alpha$ is a send action), ${\it Rcv}_{1}$ (if $\alpha$ is a receive action), or ${\it Rcv}_2$ (for receiving  $(\mathcal{C},\nrcv(\mathfrak{m}))$which are not derivable from ${\it Rcv}_1$), and
by assumption $\oc t_1\cc_\ell\simeq_{rbr}\oc t_1\cc_\ell$ implies $\oc
t_1\cc_\ell\simeq_{br}\oc t_1\cc_\ell$. Thus the transfer conditions
of Definition~\ref{Def::rbrbism} hold.

\noindent \textbf{Case~\ref{case::2}}. Every transition $\oc
t_1+t_1'\cc_\ell\overto{(\mathcal{C},\eta)} t$ owes to $\oc
t_1\cc_\ell\overto{(\mathcal{C},\eta)} t$ or $\oc
t_1'\cc_\ell\overto{(\mathcal{C},\eta)} t$ by application of ${\it Choice}$, or is implied by application of ${\it Rcv}_2$, i.e., $\oc
t_1+t_1'\cc_\ell\overto{(\mathcal{C},\nrcv(\mathfrak{m}))} \oc
t_1+t_1'\cc_\ell$  iff there exists no $\oc t_1+t_1'\cc_\ell\overto{(\mathcal{C}',\nrcv(\mathfrak{m}))}t$ for some $t$ such that $\mathcal{C}\preccurlyeq \mathcal{C}'$. For the former case, $\oc
t_1\cc_\ell\simeq_{rbr}\oc t_2\cc_\ell$ and $\oc
t_1'\cc_\ell\simeq_{rbr}\oc t_2'\cc_\ell$ imply there is a $ t'$
such that $\oc t_2\cc_\ell\overto{(\langle\mathcal{C},\eta)\rangle} t'$ or $\oc
t_2'\cc_\ell\overto{\langle(\mathcal{C},\eta)\rangle} t'$ and
$ t\simeq_{br} t'$. Thus $\oc
t_2+t_2'\cc_\ell\overto{\langle(\mathcal{C},\eta)\rangle}  t'$ with
$ t\simeq_{br} t'$. For the latter case by $\it Choice$, there exists no $\oc t_1\cc_\ell\overto{(\mathcal{C}',\nrcv(\mathfrak{m}))}t$ and $\oc t_1'\cc_\ell\overto{(\mathcal{C}',\nrcv(\mathfrak{m}))}t$ for some $t$ such that $\mathcal{C}\preccurlyeq \mathcal{C}'$. Thus by ${\it Rcv}_2$, $\oc t_1\cc_\ell\overto{(\mathcal{C},\nrcv(\mathfrak{m}))}\oc t_1\cc_\ell$ and $\oc t_1'\cc_\ell\overto{(\mathcal{C},\nrcv(\mathfrak{m}))}\oc t_1'\cc_\ell$. We remark that transitions derived by application of ${\it Rcv}_2$ are those that cannot be derived from ${\it Rcv}_1$. The greatest value of the network constraints of such transitions either have the disconnectivity pair in the form of $?\nconn \ell$ or have no connectivity pair in the form of $?\conn \ell$. 
This implies that such transitions can not be mimicked by application of ${\it Rcv}_1$ (since it will add constraints of the form $?\conn\ell$) .
Therefore, $\oc
t_1\cc_\ell\simeq_{rbr}\oc t_2\cc_\ell$ and $\oc
t_1'\cc_\ell\simeq_{rbr}\oc t_2'\cc_\ell$ imply that $\oc t_2\cc_\ell\overto{(\mathcal{C},\nrcv(\mathfrak{m}))}\oc t_2\cc_\ell$ and $\oc t_2'\cc_\ell\overto{(\mathcal{C},\nrcv(\mathfrak{m}))}\oc t_2'\cc_\ell$ which are derived by application of ${\it Rcv}_2$. Consequently $\oc t_2+t_2'\cc_\ell\overto{(\mathcal{C},\nrcv(\mathfrak{m}))}\oc t_2+t_2'\cc_\ell$.

\noindent \textbf{Case~\ref{case::13}}
Suppose that $\ell :t: t_1\overto{(\mathcal{C}^\ast,\eta)}t^\ast$, then three cases can be distinguished: \begin{itemize}
\item  It owes to $t_1\overto{(\mathcal{C},\snd(\mathfrak{m}))}t_1'$ by application of ${\it Inter}_1'$, and $\mathcal{C}^\ast=\mathcal{C}[\ell/?]$, $\eta=\nsnd(\mathfrak{m},\ell)$ and $t^\ast=\oc t_1'\cc_\ell$. By application of $\it Snd$, it implies that $\oc t_1\cc_\ell\overto{(\mathcal{C}[\ell/?],\nsnd(\mathfrak{m},\ell))}\oc t_1'\cc_\ell$. By assumption $\oc t_1\cc_\ell\simeq_{rbr}\oc t_2\cc_\ell$ implies that  $\oc t_2\cc_\ell\overto{ (\mathcal{C}[\ell/?],\nsnd(\mathfrak{m},\ell))}\oc t_2'\cc_\ell$ and $\oc t_1'\cc_\ell\simeq_{br}\oc t_2'\cc_\ell$. Therefore, by rule $\it Snd$, $t_2\overto{(\mathcal{C},\snd(\mathfrak{m}))}t_2'$, and hence by application of ${\it Inter}_1'$, $\ell :t: t_2 \overto{(\mathcal{C}[\ell/?],\nsnd(\mathfrak{m},\ell))}\oc t_2'\cc_\ell$.
\item It owes to $t_1\overto{(\mathcal{C},\rcv(\mathfrak{m}))}t_1'$ by application of either ${\it Inter}_2'$ or ${\it Inter}_3'$. This case is proved with the same argumentation as the previous case.
\item If $t_1$ and $t_2$ are of the form ${\it sense}(\ell',t_1^\ast,t_1^{\ast\ast})$ and ${\it sense}(\ell',t_2^\ast,t_2^{\ast\ast})$ respectively, and the transition owes to either ${\it Sen}_3$ or ${\it Sen}_4$. Assume it was derived by ${\it Sen}_3$, as the other case can be proved with the same argumentation. Thus, $t_1^\ast\noverto{\rcv(\mathfrak{m})}$, $t_1^{\ast\ast}\overto{(\mathcal{C},\rcv(\mathfrak{m}))}{t_1^{\ast\ast}}'$, $\mathcal{C}^\ast=\{\ell'\conn \ell\}\cup\mathcal{C}[\ell/?]$, $\eta=\nrcv(\mathfrak{m})$ and $t^\ast=t$. Therefore, by application of ${\it Rcv}_2$, $\oc t_1\cc_\ell\overto{(\mathcal{C}^\ast,\nrcv(\mathfrak{m}))}\oc t_1\cc_\ell$, and by application of ${\it Sen}_2$ and ${\it Rcv}_1$, $\oc t_1\cc_\ell\overto{(\{\ell'\nconn \ell\}\cup\mathcal{C}[\ell/?],\nrcv(\mathfrak{m}))}\oc {t_1^{\ast\ast}}'\cc_\ell$. By assumption $\oc t_1\cc_\ell\simeq_{rbr}\oc t_2\cc_\ell$ implies that $\oc t_2\cc_\ell\overto{(\mathcal{C}^\ast,\nrcv(\mathfrak{m}))}\oc t_2\cc_\ell$ and $\oc t_2\cc_\ell\overto{(\{\ell'\nconn \ell\}\cup\mathcal{C}[\ell/?],\nrcv(\mathfrak{m}))}\oc {t_2^{\ast\ast}}'\cc_\ell$ where $\oc {t_1^{\ast\ast}}'\cc_\ell\simeq_{br}\oc {t_2^{\ast\ast}}'\cc_\ell$. Thus, $t_2^\ast\noverto{\rcv(\mathfrak{m})}$, $t_2^{\ast\ast}\overto{(\mathcal{C},\rcv(\mathfrak{m}))}{t_2^{\ast\ast}}'$ (as the only way to generate the pair $\ell'\nconn \ell$ is through the $\it sense$ operator) and $\ell:t:t_2\overto{(\mathcal{C}^\ast,\nrcv(\mathfrak{m}))}t$ by application of ${\it Sen}_3$.
\end{itemize}

\noindent \textbf{Case~\ref{case::6}}.~We prove that if
$ t_1\simeq_{br} t_2$ then
$(\nu\ell) t_1\simeq_{br}(\nu\ell) t_2$. Let
$ t_1\simeq_{br}  t_2$ be witnessed by the
branching reliable computed network bisimulation relation $\mathcal{R}$.
We define
$\mathcal{R}'=\{((\nu\ell) t_1',(\nu\ell) t_2')|( t_1', t_2')\in
\mathcal{R}\}$. We prove that $\mathcal{R}'$ is a branching reliable computed
network bisimulation relation. Suppose
$(\nu\ell) t_1'\overto{(\mathcal{C}',\eta')}(\nu\ell) t_1''$
results from the application of $\it Hid$ on
$ t_1'\overto{(\mathcal{C},\eta)} t_1''$. Since
$( t_1', t_2')\in \mathcal{R}$, there are two
cases; in the first case $\eta$ is a $\tau$ action and
$( t_1'', t_2')\in\mathcal{R}$, consequently
$((\nu\ell) t_1'',(\nu\ell) t_2')\in
\mathcal{R}'$. In second case there are $ s_1''',\ldots s_k'''$ and
$ s_1'',\ldots,s_k''$ for some $k>0$ such that $\forall i\le k\, (t_2'\Rightarrow
s_i'''\overto{\langle(\mathcal{C}_i,\eta)\rangle}  s_i''$ with
$( t_1', s_i'''),( t_1'', s_i'')\in
\mathcal{R})$, and $\langle \mathcal{C}_1\rangle \ldots,\langle \mathcal{C}_k\rangle$ is a partitioning of $\langle \mathcal{C}\rangle$. By application of $\it Hid$,
$\forall i\le k\,((\nu\ell) t_2'\Rightarrow (\nu\ell) s_i'''$ with
$((\nu\ell) t_1',(\nu\ell) s_i''')\in
\mathcal{R}')$. There are two cases to consider:
\begin{itemize}
\item $\langle(\mathcal{C}_i,\eta)\rangle=(\mathcal{C}_i,\eta)$: Consequently
$(\nu\ell) s_i'''\overto{(\mathcal{C}_i',\eta')}
(\nu\ell) s_i''$ where $(\mathcal{C}_i',\eta')=(\mathcal{C}_i,\eta)[?/\ell]$.
\item $\langle(\mathcal{C}_i,\eta)\rangle\neq(\mathcal{C}_i,\eta)$: in this case $\eta$ is of the form
$\nsnd(\mathfrak{m},?)$, $\eta'=\eta$, and $\mathcal{C}'_i=\mathcal{C}_i[?/\ell]$. If
$\langle(\mathcal{C}_i,\eta)\rangle= (\mathcal{C}_i,\eta)[\ell/?]$ then $\langle(\mathcal{C}_i,\eta)\rangle[?/\ell]=(\mathcal{C}_i',\eta')$ holds, otherwise $\langle(\mathcal{C}_i,\eta)\rangle=(\mathcal{C}_i,\eta)[\ell'/?]$, where $\ell'\neq
\ell$, and hence $\langle(\mathcal{C}_i,\eta)\rangle[?/\ell]$ is a counterpart of $(\mathcal{C}_i',\eta')$. Consequently $(\nu\ell) s_i'''\overto{\langle(\mathcal{C}_i',\eta')\rangle}
(\nu\ell) s_i''$.
\end{itemize} Owing to the fact that a subset of $\mathcal{C}_1[?/\ell],\ldots,\mathcal{C}_k[?/\ell]$ constitutes a partitioning of $\mathcal{C}[\ell/?]$, and according to the discussion above, there are $s_1''',\ldots, s_j'''$ and
$ s_1'',\ldots,s_j''$ for some $j\le k$ such that $\forall i\le j,\,(\nu\ell) t_2'\Rightarrow
(\nu\ell) s_i'''\overto{\langle(\mathcal{C}_i',\eta')\rangle}
(\nu\ell) s_i''$ with
$((\nu\ell) t_1',(\nu\ell) s_i'''),((\nu\ell) t_1'',(\nu\ell) s_i'')\in
\mathcal{R}')$, and $\langle \mathcal{C}_1'\rangle ,\ldots,\langle \mathcal{C}_j'\rangle $ is a partitioning of $\langle \mathcal{C}'\rangle$.

Likewise we can prove that $ t_1\simeq_{rbr}  t_2$
implies $(\nu\ell) t_1\simeq_{rbr} (\nu\ell) t_2$.
To this aim we examine the root condition in
Definition~\ref{Def::rbrbism}. Suppose
$(\nu\ell) t_1\overto{(\mathcal{C}',\eta')}(\nu\ell) t_1'$.
With the same argument as above,
$(\nu\ell) t_2\overto{\langle(\mathcal{C}',\eta')\rangle}(\nu\ell) t_2'$.
Since $ t_1'\simeq_{br}  t_2'$, we proved that
$(\nu\ell) t_1'\simeq_{br} (\nu\ell) t_2'$.
Concluding $(\nu\ell) t_1\simeq_{rbr}
(\nu\ell) t_2$.

\noindent \textbf{Case~\ref{case::9}}.~From the three remaining
cases, we focus on the most challenging case, which is the
communication merge operator $|$, as the other operators are proved
in a similar way.
First we prove that if $ t_1\simeq_{br}
t_2$, then $ t_1\parallel  t\simeq_{br}
t_2\parallel  t$. Let $ t_1\simeq_{br}
t_2$ be witnessed by the branching reliable computed network
bisimulation relation $\mathcal{R}$. We define
$\mathcal{R}'=\{( t_1'\parallel  t',
t_2'\parallel
t')\mid ( t_1', t_2')\in \mathcal{R},\mbox{
$ t'$ any computed network term}\}$. We prove that
$\mathcal{R}'$ is a branching reliable computed network bisimulation
relation. Suppose $ t_1'\parallel
t\overto{(\mathcal{C}^\ast,\eta)} t^\ast$. There are
several cases to consider:
\begin{itemize}
\item Suppose $\eta$ is of the form $\nsnd(\mathfrak{m},\ell)$. First let it be performed by $ t_1'$, and $ t$ participated in the communication. That is, $ t_1'\overto{(\mathcal{C}_1,\nsnd(\mathfrak{m},\ell))} t_1''$ and $ t\overto{(\mathcal{C},\nrcv(\mathfrak{m}))} t'$ give rise to the transition $ t_1'\parallel  t\overto{(\mathcal{C}_1\cup C[\ell/?],\nsnd(\mathfrak{m},\ell))} t_1''\parallel  t'$. As $( t_1', t_2')\in \mathcal{R}$ and
$ t_1'\overto{(\mathcal{C}_1,\nsnd(\mathfrak{m},\ell))} t_1''$, there
are $ s_1''',\ldots,s_k'''$ and $s_1'',\ldots,s_k''$ for some $k>0$ such that $\forall i\le k\,( t_2'\Rightarrow s_i'''\overto{(\mathcal{C}_{1_i}[\ell'/\ell],\nsnd(\mathfrak{m},\ell'))}s_i''$,
where $(\ell=? \vee \ell=\ell')$, with $( t_1', s_i'''),( t_1'', s_i'')\in \mathcal{R})$, and $\mathcal{C}_{1_1}[\ell'/\ell],\ldots,\mathcal{C}_{1_k}[\ell'/\ell]$ is a partitioning of $\mathcal{C}_1[\ell'/\ell]$. Hence $\forall i\le k\,( t_2'\parallel t\Rightarrow  s_{i}'''\parallel  t\overto{((\mathcal{C}_{1_i}[\ell'/\ell]\cup C)[\ell'/?],\nsnd(\mathfrak{m},\ell'))} s_i''\parallel  t'$ with $( t_1'\parallel  t, s_i'''\parallel  t),( t_1''\parallel  t', s_i''\parallel  t')\in \mathcal{R}')$, and $(\mathcal{C}_{1_1}[\ell'/\ell]\cup\mathcal{C})[\ell'/?],\ldots,(\mathcal{C}_{1_k}[\ell'/\ell]\cup\mathcal{C})[\ell'/?]$ is a partitioning of $(\mathcal{C}_1[\ell'/\ell]\cup\mathcal{C})[\ell'/?]$.

Now suppose that the send action was performed by $ t$, and $ t_1'$ participated in the communication. That is, $ t_1'\overto{(\mathcal{C}_1,\nrcv(\mathfrak{m}))} t_1''$ and $ t\overto{(\mathcal{C},\nsnd(\mathfrak{m},\ell))} t'$
give rise to the transition $ t_1'\parallel  t\overto{(\mathcal{C}_1\cup \mathcal{C}[\ell/?],\nsnd(\mathfrak{m},\ell))} t_1''\parallel  t'$. Since $( t_1', t_2')\in \mathcal{R}$ and $ t_1'\overto{(\mathcal{C}_1,\nrcv(\mathfrak{m}))} t_1''$, there are $ s_1''',\ldots,s_k'''$ and $s_1'',\ldots,s_k''$ for some $k>0$
such that
$\forall i\le k\,( t_2'\Rightarrow s_i'''\overto{(\mathcal{C}_{1_i},\nrcv(\mathfrak{m}))} s_i''$ with $( t_1', s_i'''),( t_1'', s_i'')\in \mathcal{R})$, and $\mathcal{C}_{1_1},\ldots,\mathcal{C}_{1_k}$ is a partitioning of $\mathcal{C}_1$. Therefore, $\forall i\le k\,( t_2'\parallel
t\Rightarrow  s_i'''\parallel  t\overto{(\mathcal{C}_{1_i} \cup C[\ell/?],\nsnd(\mathfrak{m},\ell))} s_i''\parallel  t'$, and $( t_1'\parallel  t, s_i'''\parallel  t),( t_1''\parallel  t', s_i''\parallel  t')\in \mathcal{R}'$) and $\mathcal{C}_{1_1} \cup \mathcal{C}[\ell/?],\ldots,\mathcal{C}_{1_k} \cup \mathcal{C}[\ell/?]$ constitute a partitioning of $\mathcal{C}_{1} \cup \mathcal{C}[\ell/?]$.

\item The case where $\eta$ is a receive action is proved in a similar way to the previous case.
\item Suppose $\eta$ is a $\tau$ action. Assume it originates from $ t_1$ by application of $\it Par$. Thus $t_1'\overto{(\mathcal{C},\tau)}t_1''$ and $(t_1',t_2')\in\mathcal{R}$ implies: either $(t_1'',t_2')\in\mathcal{R}$ and consequently $(t_1''\parallel t,t_2'\parallel t)\in \mathcal{R}'$, or there are $ s_1''',\ldots,s_k'''$ and $s_1'',\ldots,s_k''$ for some $k>0$
such that
$\forall i\le k\,( t_2'\Rightarrow s_i'''\overto{(\mathcal{C}_{i},\tau)} s_i''$ with $( t_1', s_i'''),( t_1'', s_i'')\in \mathcal{R})$, and $\mathcal{C}_{1}, \ldots,\mathcal{C}_{k} $ constitute a partitioning of $\mathcal{C}$. Therefore, $\forall i\le k\,( t_2'\parallel
t\Rightarrow  s_i'''\parallel  t\overto{(C_{i}  ,\tau)} s_i''\parallel  t'$, and $( t_1'\parallel  t, s_i'''\parallel  t),( t_1''\parallel  t', s_i''\parallel  t')\in \mathcal{R}'$). The case when $t\overto{(\mathcal{C},\tau)}t'$ implies $t_1'\parallel t\overto{(\mathcal{C},\tau)} t_1'\parallel t'$ by application of $\it Par$ is straightforward.
\item The case when $\eta$ is an internal action is easy to prove (similar to the second case of the previous case).
\end{itemize}
Likewise we can prove that $ t_1\simeq_{rbr}  t_2$
implies
$ t\parallel t_1\simeq_{rbr} t\parallel t_2$.

\vspace{2mm}

Now let $ t_1\simeq_{rbr}  t_2$. To prove
$ t_1 \mid  t\simeq_{rbr}  t_2\mid t$,
we examine the root condition from Definition~\ref{Def::rbrbism}.
Suppose
$ t_1\mid t\overto{(\mathcal{C}^\ast,\nsnd(\mathfrak{m},\ell))}
t^\ast$. There are two cases to consider:
\begin{itemize}
\item This send action was performed by $ t_1$ at node $\ell$, and $ t$ participated in the communication. That is, $ t_1\overto{(\mathcal{C}_1,\nsnd(\mathfrak{m},\ell))} t_1'$ and $ t\overto{(\mathcal{C},\nrcv(\mathfrak{m}))} t'$, so that $ t_1\mid t\overto{(\mathcal{C}_1\cup \mathcal{C}[\ell/?],\nsnd(\mathfrak{m},\ell))} t_1'\parallel t'$. Since $ t_1\simeq_{rbr} t_2$, there is a $ t_2'$ such that $ t_2\overto{(\mathcal{C}_1,\nsnd(\mathfrak{m},\ell'))} t_2'$ with $(\ell=? \vee \ell=\ell')$ and $ t_1'\simeq_{br}  t_2'$. Then $ t_2\mid t\overto{(\mathcal{C}_1\cup \mathcal{C}[\ell'/?],\nsnd(\mathfrak{m},\ell))} t_2'\parallel
t'$. Since $ t_1'\simeq_{br} t_2'$, we proved that $ t_1'\parallel  t'\simeq_{br} t_2'\parallel
t'$.
\item The send action was performed by $ t$ at node $\ell$, and $ t_1$ participated in the communication. That is, $ t_1\overto{(\mathcal{C}_1,\nrcv(\mathfrak{m}))} t_1'$ and $ t\overto{(\mathcal{C},\nsnd(\mathfrak{m},\ell))} t'$, so that $ t_1\mid t\overto{(\mathcal{C}_1\cup\mathcal{C}[\ell/?],\nsnd(\mathfrak{m},\ell))} t_1'\parallel
t'$. Since $ t_1\simeq_{rbr} t_2$, there is a $ t_2'$ such that $ t_2\overto{(\mathcal{C}_1,\nrcv(\mathfrak{m}))} t_2'$ with $ t_1'\simeq_{br} t_2'$. Then $ t_2\mid t\overto{(\mathcal{C}_1\cup\mathcal{C}[\ell/?],\nsnd(\mathfrak{m},\ell))} t_2'\parallel
t'$. Since
$ t_1'\simeq_{br} t_2'$, we have
$ t_1'\parallel
t'\simeq_{br} t_2'\parallel
t'$.
\end{itemize}
Finally, the case where $ t_1\mid t\overto{(\mathcal{C}^\ast,\nrcv(\mathfrak{m}))} t^\ast$ can be easily dealt with. This receive action was performed by both $ t_1$ and $ t$.

Concluding, $ t_1 \mid
t\simeq_{rbr} t_2\mid t$. Likewise it can be
argued that $ t \mid
t_1\simeq_{rbr} t\mid t_2$.

\noindent \textbf{Case~\ref{case::10}}.~We prove that if
$ t_1\simeq_{br} t_2$, then
$\partial_M( t_1)\simeq_{br}\partial_M( t_2)$. Let
$ t_1\simeq_{br}  t_2$ be witnessed by the
branching reliable computed network bisimulation relation $\mathcal{R}$. We define
$\mathcal{R}'=\{(\partial_M( t_1'),\partial_M( t_2'))\mid ( t_1', t_2')\in
\mathcal{R}\}$. We prove that $\mathcal{R}'$ is a branching reliable computed
network bisimulation relation. Suppose that
$\partial_M( t_1')\overto{(\mathcal{C},\eta)}\partial_M( t_1'')$
results from the application of $\it Encap$ on
$ t_1'\overto{(\mathcal{C},\eta)}  t_1''$ such that $\eta\neq\nrcv(\mathfrak{m})\vee {\it isType}_m(\mathfrak{m})=F$. Since
$( t_1', t_2')\in \mathcal{R}$, two cases can be considered: either $\eta $ is a $\tau$ action and $(t_1'',t_2')\in\mathcal{R}$, or there are
$ s_1''',\ldots,s_k'''$ and $s_1'',\ldots, s_k''$ for some $k>0$ such that
$\forall i\le k\,( t_2'\Rightarrow
s_i'''\overto{\langle(\mathcal{C}_i,\eta)\rangle} s_i''$ with
$( t_1', s_i'''),( t_1'', s_i'')\in
\mathcal{R}$) and $\langle \mathcal{C}_1\rangle,\ldots,\langle \mathcal{C}_k\rangle$ is a partitioning of $\langle \mathcal{C}\rangle$. In the former case, $(\partial_M( t_1''),\partial_M( t_2'))\in
\mathcal{R}'$. In the latter case, by application of $\it Par$ and $\it Encap$,
$\forall i\le k\,(\partial_M( t_2')\Rightarrow
\partial_M( s_i''')\overto{\langle(\mathcal{C}_i,\eta)\rangle}\partial_M( t_2'')$ with
$(\partial_M( t_1'),\partial_M( s_i''')),(\partial_M( t_1''),\partial_M( s_i''))\in
\mathcal{R}')$.

Likewise we can prove that $ t_1\simeq_{rbr}  t_2$
implies $\partial_M( t_1)\simeq_{rbr}
\partial_M( t_2)$. To this aim we examine the root condition
in Definition~\ref{Def::rbrbism}. Suppose
$\partial_M( t_1)\overto{(\mathcal{C},\eta)}\partial_M( t_1')$.
With the same argument as above,
$\partial_M( t_2)\overto{\langle(\mathcal{C},\eta)\rangle}\partial_M( t_2')$.
Since $ t_1'\simeq_{br}  t_2'$, we proved that
$\partial_M( t_1')\simeq_{br} \partial_M( t_2')$.
Concluding $\partial_M( t_1)\simeq_{rbr}
\partial_M( t_2)$.

\noindent \textbf{Case~\ref{case::12}}~Suppose that $\mathcal{C}\rhd t_1\overto{(\mathcal{C}'\cup \mathcal{C},\eta)}t_1'$ by application of $\it TR$ since $t_1\overto{(\mathcal{C}',\eta)}t_1'$. By assumption $t_1\simeq_{rbr}t_2$ implies that $t_2\overto{(\mathcal{C}',\eta)}t_2'$ and $t_1'\simeq_{br}t_2'$. Therefore, by application of $\it TR$, $\mathcal{C}\rhd t_2\overto{(\mathcal{C}'\cup \mathcal{C},\eta)}t_2'$, and $t_1'\simeq_{br}t_2'$ concludes that $\mathcal{C}\rhd t_1\simeq_{rbr}\mathcal{C}\rhd t_2$.

\end{proof}
\section{Soundness of {\it RCNT} axiomatization}\label{sec::soundness}
As two rooted branching computed network bisimilar terms are also rooted branching reliable computed network bisimilar, the soundness of axioms which are in common with the lossy setting are established \cite{FatemehFI10}. Thus, to prove Theorem \ref{The::soundnessCNT}, it suffices to prove the soundness of each new axiom in comparison with the lossy setting, i.e., ${\it Dep}_{0-7}$, ${\it TRes}_{1-5}$, ${\it LM}_{1,2}'$, and $T_1$, modulo rooted branching reliable computed network bisimilarity.

We focus on the soundness of ${\it Dep}_{0}$ and $T_1$, as the soundness of the remaining axioms can be argued in a similar fashion. To prove ${\it Dep}_{0}$, we show that both sides of the axiom satisfy the transfer conditions of Definition \ref{Def::rbrbism}. 
Three cases can be distinguished. In following cases, for the sake of brevity, we write $X$ for $\rec\mathfrak{Q}\cdot\sum_{\mathfrak{m}'\not\in{\it Message}(t,\emptyset)}(\{\},\nrcv(\mathfrak{m}')).\mathfrak{Q}+\ell:\mathfrak{Q}:t$:
\begin{enumerate}
\item Assume $\oc t\cc_\ell\overto{(\mathcal{C}[\ell/?],\nsnd(\mathfrak{m},\ell))}\oc t'\cc_\ell$ since $t\overto{\mathcal{C},\snd(\mathfrak{m})}t'$ by application of $\it Snd$. By application of ${\it Inter}_1'$, $\ell:X:t\overto{(\mathcal{C}[\ell/?],\nsnd(\mathfrak{m},\ell))}\oc t'\cc_\ell$. Then, by application of $\it Rec$ and $\it Choice$, $X\overto{(\mathcal{C}[\ell/?],\nsnd(\mathfrak{m},\ell))}\oc t'\cc_\ell$.
\item  Assume $\oc t\cc_\ell\overto{(\mathcal{C}[\ell/?]\cup\{?\conn \ell\},\nrcv(\mathfrak{m}))}\oc t'\cc_\ell$ since $t\overto{\mathcal{C},\rcv(\mathfrak{m})}t'$ by application of ${\it Rcv}_1$. Thus by application of ${\it Inter}_2'$, $\ell:X:t\overto{(\mathcal{C}[\ell/?]\cup\{?\conn \ell\},\nrcv(\mathfrak{m}))}\oc t'\cc_\ell$.  Then, by application of $\it Rec$ and $\it Choice$, $X\overto{(\mathcal{C}[\ell/?]\cup\{?\conn \ell\},\nrcv(\mathfrak{m}))}\oc t'\cc_\ell$.
\item Assume $\oc t\cc_\ell\overto{(\mathcal{C},\nrcv(\mathfrak{m}))}\oc t\cc_\ell$ since $\oc t\cc_\ell\noverto{(\mathcal{C},\rcv(\mathfrak{m})}$ and $\nexists \mathcal{C}'(\oc t\cc_\ell \noverto{(\mathcal{C}',~\nrcv(\mathfrak{m}))}\,\wedge \, \mathcal{C}\preccurlyeq \mathcal{C}')$ by application of ${\it Rcv}_2$. Two cases can be distinguished: \begin{itemize}
\item Assume $t\noverto{\rcv(\mathfrak{m})}$, and consequently $m\in{\it Message}(t)$. Thus, by application of $\it Rec$, $\it Choice$ and $\it Prefix$, $X\overto{(\mathcal{C},\nrcv(\mathfrak{m})}X$, where $\mathcal{C}=\{\}$.
\item $t\overto{(\mathcal{C}'',\rcv(\mathfrak{m}))}t'$ for some $t'$, and consequently $m\not\in{\it Message}(t)$. Thus, the assumption $\oc t\cc_\ell\noverto{(\mathcal{C},\rcv(\mathfrak{m})}$ and $\nexists \mathcal{C}'(\oc t\cc_\ell \noverto{(\mathcal{C}',~\nrcv(\mathfrak{m}))}\,\wedge \, \mathcal{C}\preccurlyeq \mathcal{C}')$ imply that $?\nconn \ell\in\mathcal{C}$ while $?\conn \ell\in\mathcal{C}''$ due to application of ${\it Rcv}_1$. Then by application of ${\it Inter}_3'$, $\ell:X:t\overto{(\mathcal{C}[\ell/?]\cup\{?\nconn \ell\},\nrcv(\mathfrak{m}))}\oc t'\cc_\ell$.  Hence, $X\overto{(\mathcal{C}[\ell/?]\cup\{?\nconn \ell\},\nrcv(\mathfrak{m}))}\oc t'\cc_\ell$ by application of $\it Rec$ and $\it Choice$.
\end{itemize}
\end{enumerate}
We focus on the soundness of $T_1$. The only transition that the terms $(\mathcal{C}',\eta).((\mathcal{C}_1,\eta).t+(\mathcal{C}_2,\eta).t+t')$ and $(\mathcal{C}',\eta).((\mathcal{C},\eta).t+t')$ in $T_1$ can do is $\overto{(\mathcal{C}',\eta)}$ and the resulting terms $(\mathcal{C}_1,\eta).t+(\mathcal{C}_2,\eta).t+t'$ and $(\mathcal{C},\eta).t+t'$ are branching reliable computed network bisimilar, witnessed by the relation $\mathcal{R}$ constructed as follows:\[\mathcal{R}=\{((\mathcal{C}_1,\eta).t+(\mathcal{C}_2,\eta).t+t',(\mathcal{C},\eta).t+t'),(t,t)\,\mid \, t\in{\it RCNT}\}\]The pair $((\mathcal{C}_1,\eta).t+(\mathcal{C}_2,\eta).t+t',(\mathcal{C},\eta).t+t')$ satisfies the transfer conditions of Definition \ref{Def::brbism}. Because every initial transition that $(\mathcal{C}_1,\eta).t+(\mathcal{C}_2,\eta).t+t'$ can perform owing to $t'$, $(\mathcal{C},\eta).t+t'$ can perform too. If $(\mathcal{C}_1,\eta).t+(\mathcal{C}_2,\eta).t+t'$ can perform a $(\mathcal{C}_1,\eta)$ or $(\mathcal{C}_2,\eta)$-transition, $ (\mathcal{C},\eta).t+t'$ can also perform it by application of $\it Exe$. Vice versa, if $(\mathcal{C},\eta).t+t'$ can perform a $(\mathcal{C},\eta)$-transition, then as $\mathcal{C}_1$ and $\mathcal{C}_2$ form a partitioning of $\mathcal{C}$, $(\mathcal{C}_1,\eta).t+(\mathcal{C}_2,\eta).t+t'$ can perform a corresponding $(\mathcal{C}_1,\eta)$- or $(\mathcal{C}_2,\eta)$-transition.
\section{Completeness of  {\it RCNT} axiomatization}\label{sec::completeness}
To define {\it RCNT} terms with a finite-state behavior, we borrow the syntactical restriction of \cite{FatemehFI10} on recursive terms $\rec\mathfrak{A}\cdot t$, following the approach of~\cite{BaetenBravetti}. We consider so-called \textit{finite-state Reliable Computed Network Theory} ({\it RCNT$_f$}), obtained by
restricting recursive terms $\rec \mathfrak{A}\cdot t$ to those that of which the bound network names do not occur in the scope of parallel, communication merge, left merge, hide, encapsulation and abstraction operators in $t$.

We follow the corresponding proof of \cite{FatemehFI10} to prove Theorem \ref{The::completenessCNT} by performing the following steps:
\begin{enumerate}
\item first we show that each {\it RCNT}$_f$ term can be turned into a
\textit{normal form} consisting of only $0, (\mathcal{C},\eta).t',t'+t''$ and
$\rec\mathfrak{A}\cdot t'$, where $\mathfrak{A}$ is
guarded in $t'$;
\item next we define \textit{recursive network specification}s and prove that each guarded recursive network specification has a
unique solution;
\item finally we show that our axiomatization is ground-complete for normal forms, by showing that equivalent normal forms are
solutions for the same guarded recursive network specification.
\end{enumerate}
Completeness of our axiomatization for all {\it RCNT}$_f$ terms
results from the steps $1$ and $3$. We only discuss the first step, as others are exactly the same as in the lossy setting.

\begin{proposition}\label{pro::normalize}
Each closed term $t$ of {\it RCNT}$_f$ whose network names do not occur in the scope of one of the operators $\parallel,\lm,\mid,(\nu\ell),\tau_M$ or
$\partial_M$ for some $\ell\in{\it Loc}$ and $M\subseteq{\it Msg}$, can be turned into a normal form.
\end{proposition}

We prove this by structural induction over the syntax of terms $t$ (possibly open). The base cases of induction for $t\equiv 0$ or $t\equiv \mathfrak{A}$ are trivial because they are in normal form already. The inductive cases of the induction are the following ones:
\begin{itemize}
\item if $t\equiv \oc 0\cc_\ell$, then by application of ${\it Dep}_{0,4}$ and ${\it Ch}_1$ we have $t=\rec\mathfrak{Q}\cdot\sum_{\mathfrak{m}'\not\in{\it Msg}}(\{\},\nrcv(\mathfrak{m}')).\mathfrak{Q} $, which is in normal form.
\item if $t\equiv \oc \alpha.t'\cc_\ell$ or $t\equiv \oc t'+t''\cc_\ell$ or $\oc{\it sense}(\ell',t',t'')\cc_\ell$ or $\oc \mathfrak{A}\cc_\ell$, then $t$ can be
turned into a normal form by application of
axioms ${\it Dep}_{0-5,6,7}$ and induction over $\oc t'\cc_\ell$ and
$\oc t''\cc_\ell$.
\item if $t\equiv (\mathcal{C},\eta).t'$ or $t\equiv t'+t''$, then $t$ can be
turned into normal form by induction over $t'$ and $t''$.
\item the other cases can be treated in the same way as in \cite{FatemehFI10}.
\end{itemize}

\section{Proofs of Section \ref{sec::precon}}\label{sec::preorder}
We first prove Theorem \ref{Th::preorder} which indicates that the refinement relation is a preorder relation and has the precongruence property, and then we discuss the proof of Proposition \ref{Pro::rules}.

\subsection{Proof of Theorem \ref{Th::preorder}}
We first show that the refinement relation is a preorder relation and then discuss its precongruence property. To prove that refinement is a preorder, we must show that it is reflexive and transitive. As it is trivial that Definition \ref{Def::refine} is reflexive, we focus on its transitivity property.

Regrading the well-formedness conditions imposed on {\it RCNT} terms, the transitivity property of our refinement relation, i.e., $t_1\sqsubseteq t_2$ and $t_2\sqsubseteq s$ implies that $t_1\sqsubseteq s$, can be only proved when $t_1$ and $t_2$ have no prefixed-actions with a multi-hop network constraint. For such terms, Definition \ref{Def::refine} enforces they mimic the behavior of each other by the first and third conditions. In other words, for reliable computed network terms with no prefixed-actions with multi-hop network constraints, refinement and strong bisimulation of~\cite{plotkin81} coincide.

\begin{lemma}[Transitive property]\label{Lem::trans}$t_1\sqsubseteq t_2$ and $t_2\sqsubseteq s$ implies that $t_1\sqsubseteq s$.
\end{lemma}
\begin{proof}
Assume sets of refinement relations $\mathcal{R}_\mathcal{C}^1$ and $\mathcal{R}_\mathcal{C}^2$ witnessing $t_1\sqsubseteq t_2$ and $t_2\sqsubseteq s$, respectively. We construct a set of refinement relations $\mathcal{R}_\mathcal{C}'=\{(t_1',s')\,\mid\, (t_2',s')\in \mathcal{R}_\mathcal{C}^2\,\wedge\, t_1'\,\mathcal{R}_\mathcal{C}^1\,t_2\}$ for any well-formed network constraint $\mathcal{C}$. We show that $t_1' \mathcal{R}_{\mathcal{C}}' s'$ satisfies the transfer conditions of Definition \ref{Def::refine}.

Assume $t_1'\overto{(\mathcal{C}',\eta)}t_1''$. By assumption $t_1' \mathcal{R}_{\mathcal{C}}^1 t_2'$ implies that $t_2'\overto{(\mathcal{C}',\eta)}t_2''$ such that $t_1'' \mathcal{R}_{\mathcal{C}\cup \mathcal{C}'}^1 t_2''$. By the assumption $t_2' \mathcal{R}_{\mathcal{C}}^2 s'$, there are three cases to consider:\begin{itemize}
\item $\eta=\tau$ and $t_2''~ \mathcal{R}_{\mathcal{C}\cup\mathcal{C}'}^2~ s' $ with $\mathcal{C}\cup\mathcal{C}'\models \mathcal{M}$. Thus by construction, $t_1''~ \mathcal{R}_{\mathcal{C}\cup\mathcal{C}'}'~ s'$.
\item There is an $s''$ such that $s'\overto{(\mathcal{C},\eta)}s''$, and $t_2''~ \mathcal{R}_{\mathcal{C}\cup\mathcal{C}'}^2~ s''$, and $\mathcal{C}\cup\mathcal{C}'\models \mathcal{M}$. Thus by construction, $t_1''~ \mathcal{R}_{\mathcal{C}\cup\mathcal{C}'}'~ s'$.
\item $\eta=\iota$ for some $\iota\in{\it IAct}\cup\{\tau\}$ and there is an $s''$ such that $s'\overto{(\mathcal{M},\iota)}s''$ with $t_2''~ \mathcal{R}_{\mathcal{C}\cup\mathcal{C}'}^2~ s''$. Thus by construction, $t_1''~ \mathcal{R}_{\mathcal{C}\cup\mathcal{C}'}'~ s''$.
\end{itemize}

Assume $s'\overto{(\mathcal{C}',\eta)}s''$. Hence $t_2'~ \mathcal{R}_{\mathcal{C}\cup\mathcal{C}'}^2~ s'$ implies that there is a $t_2''$ such that $t_2'\overto{(\mathcal{C}',\eta)}t_2''$ with $t_2''~ \mathcal{R}_{\mathcal{C}\cup\mathcal{C}'}^2~ s''$. By assumption $t_1' \mathcal{R}_{\mathcal{C}}^1 t_2'$ implies that $t_1'\overto{(\mathcal{C}',\eta)}t_1''$ such that $t_1'' \mathcal{R}_{\mathcal{C}\cup \mathcal{C}'}^1 t_2''$, and consequently $t_1''~ \mathcal{R}_{\mathcal{C}\cup\mathcal{C}'}'~ s''$.

Assume $s\overto{(\mathcal{M},\iota)}s'$. Therefore $t_2'~ \mathcal{R}_{\mathcal{C}\cup\mathcal{C}'}^2~ s'$ implies that there are $t_2'''$ and $t_2''$ such that $t_2'\xRightarrow{\mathcal{C}'}t_2'''\overto{(\mathcal{C}'',\iota)}t_2''$ with $t_2'''~ \mathcal{R}_{\mathcal{C}\cup\mathcal{C}'}^2~ s'$ and $t_2''~ \mathcal{R}_{\mathcal{C}\cup\mathcal{C}'\cup\mathcal{C}''}^2~ s''$. As every transition of $t_2'$ is mimicked by $t_1'$, there are $t_1'''$ and $t_1''$ such that $t_1'\xRightarrow{\mathcal{C}'}t_1'''\overto{(\mathcal{C}'',\iota)}t_1''$ with $t_1'''~ \mathcal{R}_{\mathcal{C}\cup\mathcal{C}'}^1~ t_2'''$ and $t_1''~ \mathcal{R}_{\mathcal{C}\cup\mathcal{C}'\cup\mathcal{C}''}^1~ t_2''$. Concluding, there are $t_1'''$ and $t_1''$ such that $t_1'\xRightarrow{\mathcal{C}'}t_1'''\overto{(\mathcal{C}'',\iota)}t_1''$ with $t_1'''~ \mathcal{R}_{\mathcal{C}\cup\mathcal{C}'}'~ s'$ and $t_1''~ \mathcal{R}_{\mathcal{C}\cup\mathcal{C}'\cup\mathcal{C}''}'~ s''$.
\end{proof}

\begin{theorem}\label{Th::axiprecong}Refinement is a precongruence for terms with respect to the {\it RCNT} operators.
\end{theorem}
\begin{proof}
Assume that $t_1\sqsubseteq s_1$ and $t_2\sqsubseteq s_2$. We first show that $t_1+t_2\sqsubseteq s_1+s_2$.
There are sets of refinement relations $\mathcal{R}_\mathcal{C}^1$ and $\mathcal{R}_\mathcal{C}^2$ witnessing $t_1\sqsubseteq s_1$ and $t_2\sqsubseteq s_2$, respectively. We construct a set of refinement relations $\mathcal{R}_\mathcal{C}=\mathcal{R}_\mathcal{C}^1 \cup \mathcal{R}_\mathcal{C}^2\cup\{(t_1',s_1+s_2)\,\mid\,t_1'~\mathcal{R}_\mathcal{C}^1~s_1\}\cup\{(t_2',s_1+s_2)\,\mid\,t_2'~\mathcal{R}_\mathcal{C}^2~s_2\}$ for any well-formed network constraint $\mathcal{C}$. We show that $ \mathcal{R}_{\{~\}}=\{(t_1+t_2,s_1+s_2)\}\cup\mathcal{R}_{\{~\}}^1 \cup \mathcal{R}_{\{~\}}^2 $ satisfies the transfer conditions of Definition \ref{Def::refine}.

Assume $t_1+t_2\overto{(\mathcal{C}',\eta)}t_1'$ owing to  $t_1\overto{(\mathcal{C}',\eta)}t_1'$. By assumption $t_1~ \mathcal{R}_{\{~\}}^1~s_1$. Three cases can be considered:
\begin{itemize}
\item $\eta=\tau$ and $t_1'~ \mathcal{R}_{\mathcal{C}\cup\mathcal{C}'}~ s_1 $, with $\mathcal{C}\cup\mathcal{C}'\models \mathcal{M}$. Thus by construction $t_1'~ \mathcal{R}_{\mathcal{C}\cup\mathcal{C}'}~ s_1+s_2 $.
\item There is an $s_1'$ such that $s_1\overto{(\mathcal{C},\eta)}s_1'$, and $t_1'~ \mathcal{R}_{\mathcal{C}\cup\mathcal{C}'}^1~ s_1'$, and $\mathcal{C}\cup\mathcal{C}'\models \mathcal{M}$. Thus by sos rule ${\it Choice}$, there is an $s_1'$ such that $s_1+s_2\overto{(\mathcal{M},\iota)}s_1'$ and by construction $t_1'~ \mathcal{R}_{\mathcal{C}\cup\mathcal{C}'}~ s_1' $.
\item $\eta=\iota$ for some $\iota\in{\it IAct}\cup\{\tau\}$ and there is an $s_1'$ such that $s_1\overto{(\mathcal{M},\iota)}s_1'$ with $t_1'~ \mathcal{R}_{\mathcal{C}\cup\mathcal{C}'}^1~ s_1'$. Thus by the sos rule ${\it Choice}$, there is an $s_1'$ such that $s_1+s_2\overto{(\mathcal{M},\iota)}s_1'$ and by construction $t_1'~ \mathcal{R}_{\mathcal{C}\cup\mathcal{C}'}~ s_1'$.
\end{itemize}
The same discussion holds if $t_1+t_2\overto{(\mathcal{C}',\eta)}t_2'$ owing to $t_2\overto{(\mathcal{C}',\eta)}t_2'$.

Assume $s_1+s_2\overto{(\mathcal{M},\iota)}s_1'$ owing to $s_1\overto{(\mathcal{M},\iota)}s_1'$. By assumption $t_1~\mathcal{R}_\mathcal{C}^1~s_1$ implies there are $t_1''$ and $t_1'$ such that $t_1\xRightarrow{\mathcal{C}'}t_1''\overto{(\mathcal{C}'',\iota)}t_1'$ with $t_1''~ \mathcal{R}_{\mathcal{C}\cup\mathcal{C}'}^1~ s_1$ and $t_1'~ \mathcal{R}_{\mathcal{C}\cup\mathcal{C}'\cup\mathcal{C}''}^1~ s_1'$. Consequently $t_1+t_2\xRightarrow{\mathcal{C}'}t_1''\overto{(\mathcal{C}'',\iota)}t_1'$ with $t_1''~ \mathcal{R}_{\mathcal{C}\cup\mathcal{C}'}~ s_1+s_2$ and $t_1'~ \mathcal{R}_{\mathcal{C}\cup\mathcal{C}'\cup\mathcal{C}''}~ s_1'$. The same discussion holds  when $s_1+s_2\overto{(\mathcal{M},\iota)}s_2'$ owing to $s_2\overto{(\mathcal{M},\iota)}s_2'$.

Assume $s_1+s_2\overto{(\mathcal{C},\eta)}s_1'$ owing to $s_1\overto{(\mathcal{C},\eta)}s_1'$. By assumption $t_1~\mathcal{R}_\mathcal{C}^1~s_1$ implies there is a $t_1'$ such that $t_1\overto{(\mathcal{C}',\eta)}t_1'$ with  $t_1'~ \mathcal{R}_{\mathcal{C}\cup\mathcal{C}'}^1~ s_1'$. Hence, there is a $t_1'$ such that $t_1+t_2\overto{(\mathcal{C}',\eta)}t_1'$ with  $t_1'~ \mathcal{R}_{\mathcal{C}\cup\mathcal{C}'}~ s_1'$.

The above discussions together yield $t_1+t_2\sqsubseteq s_1+s_2$.

If $s_1$ and $s_2$ have no prefixed-action with a multi-hop network constraint, then we must show the following cases:
\begin{enumerate}
\item $(\mathcal{C},\eta).t_1\sqsubseteq (\mathcal{C},\eta).t_2$;
\item $(\nu\ell).t_1\sqsubseteq (\nu\ell).t_2$;
\item $t_1\parallel t_2 \sqsubseteq s_1\parallel s_2$;
\item $t_1\lm t_2 \sqsubseteq s_1\lm s_2$;
\item $t_1\mid t_2 \sqsubseteq s_1\mid s_2$;
\item $\partial_M(t_1)\sqsubseteq \partial_M(t_2)$;
\item $\tau(t_1)\sqsubseteq \tau(t_2)$;
\item $\mathcal{C}\rhd t_1\sqsubseteq \mathcal{C}\rhd t_2$;
\end{enumerate}
The above cases result from the congruence property of strong bisimilarity.

\end{proof}

The proof of Theorem \ref{Th::preorder} is an immediate result of Lemma \ref{Lem::trans} and Theorem \ref{Th::axiprecong}.
\subsection{Proof of Proposition \ref{Pro::rules}}
First we show that $(\mathcal{C},\tau).t\sqsubseteq (\mathcal{M},\iota).s \Rightarrow \mathcal{C}\rhd t \sqsubseteq (\mathcal{M},\iota).s \,\wedge\,\mathcal{C}\models\mathcal{M}$. The only transition $(\mathcal{C},\tau).t$ can make is $(\mathcal{C},\tau).t\overto{(\mathcal{C},\tau)}t$. As $\iota\neq\tau$, according to Definition \ref{Def::refine}, $t~\mathcal{R}_\mathcal{C}~(\mathcal{M},\iota).s$. We construct $\mathcal{R}_{\{~\}}'=\mathcal{R}_\mathcal{C}$ and show that it induces $ \mathcal{C}\rhd t \sqsubseteq (\mathcal{M},\iota).s$. This is trivial as any transition $ \mathcal{C}\rhd t \overto{(\mathcal{C}\cup\mathcal{C}',\eta)}t'$ is the result of $t \overto{(\mathcal{C}',\eta)}t'$. The reverse of the rule can be argued in a similar fashion.

Now, we show that $
(\mathcal{C},\iota).t\sqsubseteq (\mathcal{M},\iota).s \Rightarrow \mathcal{C}\rhd t \sqsubseteq s $. The only transition $(\mathcal{C},\iota).t$ can make is $(\mathcal{C},\iota).t\overto{(\mathcal{C},\tau)}t$. As $\iota\neq\tau$ and $\iota\in{\it IAct}$, according to Definition \ref{Def::refine}, $t~\mathcal{R}_\mathcal{C}~s$. We construct $\mathcal{R}_{\{~\}}'=\mathcal{R}_\mathcal{C}$ and show that it induces $ \mathcal{C}\rhd t \sqsubseteq s$. This is trivial as any transition $ \mathcal{C}\rhd t \overto{(\mathcal{C}\cup\mathcal{C}',\eta)}t'$ is the result of $t \overto{(\mathcal{C}',\eta)}t'$. The reverse of the rule can be argued in a similar fashion.

\end{document}